\documentclass[runningheads,oribibl]{llncs}
\usepackage[T1]{fontenc}
\usepackage[numbers,sort&compress]{natbib}

\makeatletter
\renewcommand\@biblabel[1]{#1.}
\makeatother

\usepackage[colorlinks]{hyperref}
\usepackage{color}

\usepackage{float}
\usepackage{subcaption}
\usepackage{amsmath}
\usepackage{amssymb}
\usepackage{bbm}
\usepackage{mathtools}
\usepackage{readarray}
\usepackage{adjustbox}
\usepackage{listings}
\usepackage[lighttt]{lmodern}
\usepackage[inline]{enumitem}
\usepackage{graphicx}
\usepackage{booktabs}
\usepackage{array}
\usepackage{bm}
\usepackage{multirow}
\usepackage[font=small]{caption}
\usepackage[nameinlink,capitalise]{cleveref}

\usepackage{tikz}
\usetikzlibrary{shapes.misc}
\usetikzlibrary{arrows.meta}
\usetikzlibrary{calc}
\usetikzlibrary{decorations.pathreplacing}
\tikzset{cross/.style={cross out, draw=black, minimum size=2*(#1-\pgflinewidth),
    inner sep=0pt, outer sep=0pt},
cross/.default={1pt}
}

\usepackage[small]{complexity}

\makeatletter
\if@cref@capitalise
\crefname{property}{Property}{Properties}
\else
\crefname{property}{property}{properties}
\fi
\makeatother
\Crefname{property}{Property}{Properties}

\usepackage{thmtools,thm-restate}

\newcommand{\real}{\mathbb{R}}
\newcommand{\nat}{\mathbb{N}}
\newcommand{\posint}{\mathbb{Z}^+}
\newcommand{\indicator}[1]{\mathbbm{1}_{#1}}
\newcommand{\nth}[1]{#1^{\text{th}}}
\newcommand{\defeq}{\mathrel{\raisebox{-0.3ex}{$\overset{\text{\tiny def}}{=}$}}}
\newcommand{\condSet}[2]{\left\{#1 \;\middle|\; #2\right\}}

\newcommand{\setV}{\mathcal{V}}
\newcommand{\setE}{\mathcal{E}}
\newcommand{\setC}{\mathcal{C}}
\newcommand{\setN}{\mathcal{N}}
\newcommand{\setR}{\mathcal{R}}

\newcommand{\bmax}{\overline{b}}

\DeclareMathOperator{\da}{DA}   % default action
\DeclareMathOperator{\equil}{Equil}
\DeclareMathOperator{\best}{best}
\DeclareMathOperator{\POA}{POA}

\DeclareMathOperator*{\argmax}{arg\,max}

\newcommand{\strategyLevel}[1]{\mathcal{L}_{#1}}

\newcommand{\pp}{capability-positive}   % pro-power
\newcommand{\bwr}{max-capability-preferred}    % best-without-restriction
\newcommand{\ap}{capability-negative}   % anti-power
\newcommand{\bfr}{min-capability-preferred}    % best-full-restriction
\newcommand{\dnc}{DNC}
\newcommand{\dbncgda}{\textsc{DncDa}}
\newcommand{\dbncgdas}{\textsc{DncDaS}}
\newcommand{\gam}{GMG}
\newcommand{\aog}{alternating ordering game}
\newcommand{\wltrend}{capability preference}

\newenvironment{enuminline}
{\begin{enumerate*}[label=(\roman*),itemjoin={{; }},itemjoin*={{; and }}]}
{\end{enumerate*}}

\newcommand{\V}[1]{\bm{#1}}
\newcommand{\vs}{\V{s}}

\newcommand{\bestw}[1]{W^+_{#1}}
\newcommand{\worstw}[1]{W^-_{#1}}

\newcommand{\fd}{d(\cdot)}  % delay function
\newcommand{\frg}{r_g(\cdot)}  % delay function
\newcommand{\frm}{r_m(\cdot)}  % delay function

\makeatletter
\newcommand{\printfnsymbol}[1]{%
  \textsuperscript{\@fnsymbol{#1}}%
}
\makeatother

\newcommand{\X}{,\,\allowbreak}

\begin{document}
\title{On the Impact of Player Capability on Congestion Games}
%\author{First Author\inst{1}\orcidID{0000-1111-2222-3333} \and
%Second Author\inst{2,3}\orcidID{1111-2222-3333-4444} \and
%Third Author\inst{3}\orcidID{2222--3333-4444-5555}}
\author{Yichen Yang\thanks{Equal contribution.} \and
Kai Jia\printfnsymbol{1} \and
Martin Rinard
}
\authorrunning{Yichen Yang \and
Kai Jia \and
Martin Rinard
}
\institute{Department of Electrical Engineering and Computer Science, \\
Massachusetts Institute of Technology, USA \\
\email{\{yicheny,jiakai,rinard\}@csail.mit.edu}}
%\institute{Princeton University, Princeton NJ 08544, USA \and
%Springer Heidelberg, Tiergartenstr. 17, 69121 Heidelberg, Germany
%\institute{Anonymous Institute}
%\author{Anonymous authors}
\maketitle              % typeset the header of the contribution

\begin{abstract}

We study the impact of player capability on social welfare in congestion games.
We introduce a new game, the \emph{\textbf{D}istance-bounded \textbf{N}etwork \textbf{C}ongestion game (\dnc{})}, as the basis of our study. \dnc{} is a symmetric network congestion game with a bound on the number of edges each player can use.
We show that \dnc{} is \PLS-complete in contrast to standard symmetric network congestion games which are in \P.
To model different player capabilities, we propose using programs in a Domain-Specific Language (DSL) to compactly represent player strategies. We define a player's capability as the maximum size of the programs they can use.
We introduce two variants of \dnc{} with accompanying DSLs representing the strategy spaces.
We propose four {\em \wltrend{}} properties to characterize the impact of player
capability on social welfare at equilibrium.
We then establish necessary and sufficient conditions for the four properties in the context of our \dnc{} variants.
Finally, we study a specific game where we derive exact expressions of the social welfare in terms of the capability bound. This provides examples where the social welfare at equilibrium increases, stays the same, or decreases as players become more capable.
    \keywords{Congestion game \and Player capability \and Social welfare.}

\end{abstract}

% vim: tw=0 wrap linebreak filetype=tex foldmethod=marker foldmarker=f{{{,f}}} spell spelllang=en

\section{Introduction}

Varying player capabilities can significantly affect the outcomes of strategic
games. Developing a comprehensive understanding of how different player
capabilities affect the dynamics and overall outcomes of strategic games is
therefore an important long-term research goal in the field. Central questions
include characterizing, and ideally precisely quantifying, player capabilities,
then characterizing, and ideally precisely quantifying, how these different
player capabilities interact with different game characteristics to influence or
even fully determine individual and/or group dynamics and outcomes.

We anticipate a range of mechanisms for characterizing player capabilities, from
simple numerical parameters through to complex specifications of available
player behavior. Here we use programs in a domain-specific language (DSL) to
compactly represent player strategies. Bounding the sizes of the programs
available to the players creates a natural capability hierarchy, with more
capable players able to deploy more diverse strategies defined by larger
programs. Building on this foundation, we study the effect of increasing or
decreasing player capabilities on game outcomes such as social welfare at
equilibrium. To the best of our knowledge, this paper presents the first
systematic analysis of the effect of different player capabilities on the
outcomes of strategic games.

We focus on network congestion games \citep{ roughgarden2002unfair,
fabrikant2004complexity}. All congestion games have pure Nash equilibria (PNEs)
\citep{rosenthal1973class, nisan2007algorithmic}. Congestion games have been
applied in many areas including drug design \citep{ nikitina2018congestion},
load balancing \citep{ zhang2020stochastic}, and network design
\citep{le2020congestion}. There is a rich literature on different aspects of
congestion games including their computational characteristics
\citep{ackermann2008impact}, efficiency of equilibria \citep{
correa2005inefficiency}, and variants such as weighted congestion games \citep{
christodoulou2020existence} or games with unreliable resources \citep{
nickerl2021pure}.

We propose a new network congestion game, the \emph{\textbf{D}istance-bounded
\textbf{N}etwork \textbf{C}ongestion game (\dnc{})}, as the basis of our study.
A network congestion game consists of some players and a directed graph where
each edge is associated with a delay function. The goal of each player is to
plan a path that minimizes the delay from a source vertex to a sink vertex. The
delay of a path is the sum of the delays on the edges in the path, with the
delay on each edge depending (only) on the number of players choosing the edge.
The game is symmetric when all players share the same source and sink. \citet{
fabrikant2004complexity} shows that finding a PNE is in \P{} for symmetric games
but \PLS-complete for asymmetric ones. \dnc{} is a symmetric network congestion
game in which each player is subject to a \emph{distance bound} --- i.e., a
bound on the number of edges that a player can use.

We establish hardness results for \dnc{}. We show that with the newly introduced
distance bound, our symmetric network congestion game becomes \PLS-complete. We
also show that computing the best or worst social welfare among PNEs of
\dnc{} is \NP-hard.

We then present two games for which we define compact DSLs for player
strategies. The first game is a \dnc{} variant with \emph{\textbf{D}efault
\textbf{A}ction (\dbncgda{})}. In this game, each node has a default outgoing
edge that does not count towards the distance bound. Hence a strategy can be
compactly represented by specifying only the non-default choices. We
establish that \dbncgda{} is as hard as \dnc{}. The other game is Gold and
Mines Game (\gam), where there are gold and mine sites placed on parallel
horizontal lines, and a player uses a program to compactly describe the line
that they choose at each horizontal location. We show that \gam{} is a special
form of \dbncgda.

We propose four \emph{\wltrend} properties that characterize the impact of
player capability on social welfare. In this paper, we only consider social
welfare of pure Nash equilibria. We call a game \emph{\pp} (resp. \emph{\ap}) if
social welfare does not decrease (resp. increase) when players become more
capable. A game is \emph{\bwr} (resp. \emph{\bfr}) if the worst social welfare
when players have maximal (resp. minimal) capability is at least as good as any
social welfare when players have less (resp. more) capability. Note that \bwr{}
(resp. \bfr{}) is a weaker property than \pp{} (resp. \ap{}). Due to the
hardness of \dbncgda, we analyze a restricted version (\dbncgdas{}) where all
edges share the same delay function. We identify necessary and sufficient
conditions on the delay function for each \wltrend{} property to hold
universally for all configurations in the context of \dbncgdas{} and \gam{}
(\cref{tab:property}).

Finally, we study a specific version of \gam{} where we derive exact expressions
of the social welfare in terms of the capability bound and payoff function
parameterization. We present examples where the social welfare at equilibrium
increases, stays the same, or decreases as players become more capable.

\noindent{\bfseries Summary of contributions:}
\begin{itemize}
    \item We present a framework for quantifying varying player capabilities and
        studying how different player capabilities affect the dynamics and
        outcome of strategic games. In this framework, player strategies are
        represented as programs in a DSL, with player capability defined as the
        maximal size of available programs.

    \item We propose four \emph{\wltrend} properties to characterize the impact
        of player capability on social welfare of pure Nash equilibria:
        \emph{\pp}, \emph{\bwr}, \emph{\ap}, and \emph{\bfr}.

    \item We introduce the \emph{distance-bounded network congestion game
        (\dnc{})}, a new symmetric network congestion game in which players can
        only choose paths with a bounded number of edges. We further show that
        \dnc{} is \PLS-complete. Moreover, computing the best or worst social
        welfare among equilibria of \dnc{} is \NP-hard.

    \item We introduce two variants of \dnc{}, \dnc{} with default action
        (\dbncgda{}) and Gold and Mines Game (\gam{}), with accompanying
        DSLs to compactly represent player strategies. We then establish
        necessary and sufficient conditions for the \wltrend{} properties in
        \dbncgdas{} and \gam.
    \item We study a special version of \gam{} where we fully characterize the
        social welfare at equilibrium. This characterization provides insights
        into the factors that affect whether increasing player capability is
        beneficial or not.
\end{itemize}

\noindent {\bfseries Additional related work}
There has been research exploring the results of representing player strategies
using formal computational models. \citet{TENNENHOLTZ2004programequil} proposes
using programs to represent player strategies and analyzes program equilibrium
in a finite two-player game. \citet{Fortnow2009progequil} extends the results,
representing strategies as Turing machines. Another line of research uses
various kinds of automata to model a player's strategy in non-congestion games
\citep{ almanasra2013applications}, such as repeated prisoner's dilemma
\citep{PAPADIMITRIOU1992automaton, RUBINSTEIN1986automaton}. Automata are
typically used to model bounded rationality \citep{neyman1985bounded,
papadimitriou1994complexity} or certain learning behavior \citep{chiong2009co,
ghnemat2006automata}. \citet{ neyman1997cooperation} presents asymptotic results
on equilibrium payoff in repeated normal-form games when automaton sizes meet
certain conditions. There has also been research exploring structural strategy
spaces in congestion games. \citet{ ackermann2008complexity} considers a
player-specific network congestion game where each player has a set of forbidden
edges. \citet{ chan2016congestion} studies computing mixed Nash equilibria in a
broad class of congestion games with strategy spaces compactly described by a
set of linear constraints. Unlike our research, none of the above research
defines a hierarchy of player capabilities or characterizes the effect of the
hierarchy on game outcomes.

% vim: tw=80 filetype=tex foldmethod=marker foldmarker=f{{{,f}}} spell spelllang=en

\section{Distance-bounded network congestion game}
% f{{{

It is well-known that computing PNEs in symmetric network congestion games
belongs to \P{} while the asymmetric version is \PLS-complete \citep{
fabrikant2004complexity}. We present a symmetric network congestion game where
we limit the number of edges that a player can use. We show that this
restriction makes the problem harder --- the new game is \PLS-complete, and
finding the best or worst Nash equilibrium in terms of global social welfare is
\NP-hard.

\begin{definition}
    \label{def:dbncg}
    An instance of the \emph{\textbf{D}istance-bounded \textbf{N}etwork
    \textbf{C}ongestion game (\dnc{})} is a tuple $G=(\setV \X \setE \X \setN \X
    s \X t \X b \X (d_e)_{e \in \setE})$ where:
    \begin{itemize}
        \item $\setV$ is the set of vertices in the network.
        \item $\setE\subset \setV \times \setV$ is the set of edges in the
            network.
        \item $\setN=\{1,\, \dots,\, n\}$ is the set of players.
        \item $s\in \setV$ is the source vertex shared by all players.
        \item $t \in \setV$ is the sink vertex shared by all players.
        \item $b \in \nat$ is the bound of the path length.
        \item $d_e: \nat \mapsto \real$ is a non-decreasing delay function on
            edge $e$.
    \end{itemize}

    We also require that the network has no negative-delay cycles, i.e., for
    each cycle $\setC$, we require $\sum_{e\in\setC} \min_{i\in\setN} d_e(i)
    = \sum_{e\in\setC} d_e(1) \ge 0$.

    We only consider pure strategies (i.e., deterministic strategies) in this
    paper. The strategy space of a single player contains all $s-t$ simple paths
    whose length does not exceed $b$:
    \begin{align*}
        \strategyLevel{b} \defeq
            \condSet{(p_0,\, \ldots,\, p_k)}{
            \begin{array}{l}
                p_0=s,\,p_k=t,\, (p_i,\,p_{i+1}) \in\setE, k \le b, \\
                p_i \neq p_j \text{ for } i \neq j
            \end{array}}
    \end{align*}
\end{definition}

In a \dnc{}, as in a general congestion game, a player's goal is to minimize
their delay. Let $s_i\defeq (p_{i0},\,\cdots,\,p_{ik_i}) \in \strategyLevel{b}$
denote the strategy of player $i$ where $i \in \setN$. A \emph{strategy profile}
$\vs = (s_1, \dots, s_n) \in \strategyLevel{b}^n$ consists of strategies of all
players. Let $E_i\defeq \condSet{ (p_{ij},\,p_{i,j+1})}{0\le j< k_i}$ denote the
corresponding set of edges on the path chosen by player $i$. The \emph{load} on
an edge $e\in\setE$ is defined as the number of players that occupy this edge:
$x_e \defeq |\condSet{i}{e \in E_i}|$. The delay experienced by player $i$ is
$c_i(\vs) \defeq \sum_{ e\in E_i}d_e(x_e)$. A strategy profile $\vs$ is a
\emph{pure Nash equilibrium (PNE)} if no player can improve their delay by
unilaterally changing strategy, i.e., $\forall i \in \setN: c_i(\vs) = \min_{s'
\in \strategyLevel{b}} c_i(\vs_{-i},\,s')$. All players experience infinite
delay if the distance bound permits no feasible solution (i.e., when
$\strategyLevel{b} = \emptyset$). \emph{Social welfare} is defined as the the
negative total delay of all players where a larger welfare value means on average players experience less delay: $W(\vs) \defeq - \sum_{i\in \setN}
c_i(\vs)$.

We now present a few hardness results about \dnc{}.

\begin{lemma}
    \dnc{} belongs to \PLS.
\end{lemma}
\begin{proof}
    \dnc{} is a potential game where local minima of its potential function
    correspond to PNEs \citep{rosenthal1973class}. Clearly there
    are polynomial algorithms for finding a feasible solution or evaluating the
    potential function. We only need to show that computing the best response
    of some player $i$ given the strategies of others is in \P. For each
    $v\in\setV$, we define $f(v,\, d)$ to be the minimal delay experienced by
    player $i$ over all paths from $s$ to $v$ with length bound $d$. It can be
    recursively computed via $f(v,\, d) = \min_{u\in\setV:(u,\,v) \in\setE} (
    f(u, d-1) + d_{(u,\,v)}(x_{(u,\,v)}+1))$ where $x_{(u,\,v)}$ is the load on
    edge $(u,\,v)$ caused by other players. The best response of player $i$ is
    then $f(t,\,b)$. If there are cycles in the solution, we can remove them
    without affecting the total delay because cycles must have zero delay in the
    best response.
\end{proof}

\begin{theorem}
    \dnc{} is \PLS-complete.
    \label{thm:dbncg-pls-complete}
\end{theorem}
\begin{proof}
    We have shown that \dnc{} belongs to \PLS. Now we present a \PLS-reduction
    from a \PLS-complete game to finish the proof.

    \newcommand{\Rin}{\setR^{\text{in}}}
    \newcommand{\Rout}{\setR^{\text{out}}}
    \newcommand{\Sin}{S_i^{\text{in}}}
    \newcommand{\Sout}{S_i^{\text{out}}}
    The \emph{quadratic threshold game} \citep{ackermann2008impact} is a
    \PLS-complete game in which there are $n$ players and $n(n+1)/2$ resources.
    The resources are divided into two sets $\Rin = \condSet{r_{ij}}{1 \le i < j
    \le n}$ for all unordered pairs of players $\{i,\,j\}$ and $\Rout =
    \condSet{r_i}{i \in \setN}$. For ease of exposition, we use $r_{ij} $
    and $r_{ji}$ to denote the same resource. Player $i$ has two
    strategies: $\Sin = \condSet{r_{ij}}{j \in \setN / \{i\}}$ and $\Sout =
    \{r_i\}$.

    Extending the idea in \citet{ackermann2008impact}, we reduce from the
    quadratic threshold game to \dnc{}. To simplify our presentation, we assign
    positive integer weights to edges. Each weighted edge can be replaced by a
    chain of unit-length edges to obtain an unweighted graph.

    \begin{figure}[tb]
        \centering
        \begin{subfigure}[t]{.64\textwidth}
            \centering
            \def\figdatatxt{%
                10 9 6 ~
                8 9 12 17
            }
            \begin{adjustbox}{width=.72\textwidth}
                \readarray\figdatatxt\figdata[2,4]

\begin{tikzpicture}[
    node distance={4em},
    main/.style = {draw, circle, minimum size=2em, font=\small},
    edgelabel/.style = {midway, above, sloped, font=\small},
    jumpedge/.style = {gray, dashed}
]

    % uncomment the last few commands to recompute the coordinates
    \clip (-10.5em,-30em) rectangle (19em, 7.5em);

    \tikzset{>=Latex}

    \node[main] (11) {};
    \node[main] (21) [below of=11] {$r_{12}$};
    \node[main] (22) [right of=21] {};
    \node[main] (31) [below of=21] {$r_{13}$};
    \node[main] (32) [right of=31] {$r_{23}$};
    \node[main] (33) [right of=32] {};
    \node[main] (41) [below of=31] {$r_{14}$};
    \node[main] (42) [right of=41] {$r_{24}$};
    \node[main] (43) [right of=42] {$r_{34}$};
    \node[main] (44) [right of=43] {};

    \foreach \i in {1,...,4} {
        \node[main] (s\i) [left of=\i1] {$s_{\i}$};
        \draw[->] (s\i) to (\i1);
    }

    \foreach \i in {1,...,4} {
        \node[main] (t\i) [below of=4\i] {$t_{\i}$};
        \draw[->] (4\i) to (t\i);
    }

    \node[main] (s) at ([xshift=-5em]$(s1)!0.5!(s4)$) {$s$};
    \node[main] (t) at ([yshift=-4em]$(t1)!0.5!(t4)$) {$t$};

    \foreach \y in {2,...,4}
        \foreach \x [evaluate={\xprev=\x-1}] in {2,...,\y} {
            \pgfmathtruncatemacro{\yprev}{\y-1}
            \draw[->] (\y\xprev) to node[edgelabel]{\y} (\y\x);
        }

    \foreach \y [evaluate={\yprev=\y-1}] in {2,...,4}
        \foreach \x in {1,...,\yprev} {
            \pgfmathtruncatemacro{\yprev}{\y-1}
            \draw[->] (\yprev\x) to (\y\x);
        }

    \foreach \out/\i in {90/1,45/2,-45/3,-90/4} {
        \draw[->, out=\out, in=180] (s) to node[edgelabel]{\figdata[1,\i]} (s\i);
        \draw[->] (t\i) to (t);
    }

    \draw[->, jumpedge, out=45, in=45, looseness=1.7]
        (s1) to node[edgelabel, gray, pos=0.55]
        {$r_1;\; w_1=\figdata[2,1]$} (t1);

    \draw[->, jumpedge, out=120, in=45, looseness=3.8]
        (s2) to node[edgelabel, gray, pos=0.45]
        {$r_2;\; w_2=\figdata[2,2]$} (t2);

    \draw[->, jumpedge]
        (s3) to[out=-140, in=-90, looseness=1.4]
        node[edgelabel, gray, pos=0.5]
        {$r_3;\; w_3=\figdata[2,3]$}
        ++($(t4.east) + (9em, 0)$)
        to[out=90, in=30, looseness=2] (t3);

    \draw[->, jumpedge, out=-80, in=-30, looseness=2.3]
        (s4) to node[edgelabel, gray, pos=0.5]
        {$r_4;\; w_4=\figdata[2,4]$} (t4);

    %\draw[gray,step=5em] (-10em,-30em) grid (20em, 20em);
    %\draw [brown] (-10.5em,-30em) rectangle (19em, 7.5em);
    %\clip (-11em,-30em) rectangle (20em, 8.5em);
    %\draw [brown] (current bounding box.south west) rectangle (current bounding box.north east);
\end{tikzpicture}

% vim: tw=80 filetype=tex foldmethod=marker foldmarker=f{{{,f}}} spell spelllang=en
            \end{adjustbox}
            \caption{The graph structure}
        \end{subfigure}
        \begin{subfigure}[t]{.35\textwidth}
            \centering
            \begin{adjustbox}{width=0.7\textwidth}
                \begin{tikzpicture}[
    node distance={5em},
    main/.style = {draw, circle},
    edgelabel/.style = {midway, above, sloped, font=\large}
]

    \tikzset{>=Latex}

    \node[main,minimum size=7.5em] at (0,0) {};
    \node[main,minimum size=1em] (v0) at (-1.5em,1.5em) {};
    \node[main,minimum size=1em] (v1) at (1.5em,-1.5em) {};
    \node[] (vs0) [above of=v0] {};
    \node[] (vs1) [left of=v0] {};
    \node[] (vt0) [right of=v1] {};
    \node[] (vt1) [below of=v1] {};

    \draw[->] (vs0) to (v0);
    \draw[->] (vs1) to (v0);
    \draw[->] (v1) to (vt0);
    \draw[->] (v1) to (vt1);
    \draw[->] (v0) to node[edgelabel]{$r_{ij}$} (v1);

\end{tikzpicture}

% vim: tw=80 filetype=tex foldmethod=marker foldmarker=f{{{,f}}} spell spelllang=en
            \end{adjustbox}
            \caption{Splitting the vertex containing resource $r_{ij}$}
        \end{subfigure}
        \caption{
            The \dnc{} instance corresponding to a four-player quadratic
            threshold game. The distance bound $b=19$. Non-unit-length edges
            have labels to indicate their lengths. Dashed gray edges correspond to the
            $\Sout$ strategies.
            \label{fig:dbncg-pls-proof}
        }
        \vskip-1em
    \end{figure}
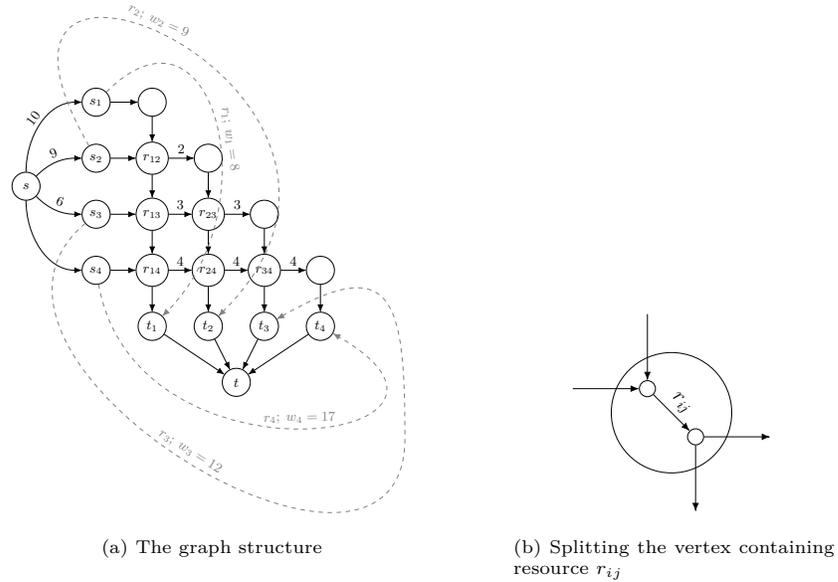

    \Cref{fig:dbncg-pls-proof} illustrates the game with four players. We create
    $n(n+1) /2$ vertices arranged as a lower triangle. We use $v_{ij}$ to denote
    the vertex at the $\nth{i}$ row (starting from top) and $\nth{j}$ column
    (starting from left) where $1 \le j \le i \le n$. The vertex $v_{ij}$ is
    connected to $v_{i, j+1}$ with an edge of length $i$ when $j < i$ and to
    $v_{i+1,j}$ with a unit-length edge when $i < n$. This design ensures that
    the shortest path from $v_{i1}$ to $v_{ni}$ is the right-down path. The
    resource $r_{ij}$ is placed at the off-diagonal vertex $v_{ij}$, which can
    be implemented by splitting the vertex into two vertices connected by a
    unit-length edge with the delay function of $r_{ij}$. Note that this implies
    visiting a vertex $v_{ij}$ incurs a distance of $1$ where $i \neq j$. We
    then create vertices $s_i$ and $t_i$ for $1 \le i \le n$ with unit-length
    edges $(s_i,\,v_{i1})$ and $(v_{ni},\,t_i)$. We connect $s_i$ to $t_i$ with
    an edge of length $w_i$, which represents the resource $r_i$. Let $b$ be the
    distance bound. We will determine the values of $w_i$ and $b$ later.
    The source $s$ is connected to $s_i$ with an edge of length $b - w_i - 1$.
    Vertices $t_i$ are connected to the sink $t$ via unit-length edges.

    We define the following delay functions for edges associated with $s$ or $t$:
    \begin{align*}
        d_{(s,\,s_i)}(x) &= \indicator{x\ge2} \cdot (|\setN|+1)R \quad\quad
        d_{(t_i,\,t)}(x) = (|\setN|-i)R \\
        \text{where } R &= \left(\sum_{r \in \Rin \cup \Rout}
            \max_{i\in\setN} d_r(i)\right) + 1
    \end{align*}

    We argue that player $i$ chooses edges $(s,\,s_i)$ and $(t_i,\,t)$ in
    their best responses. Since $R$ is greater than the maximum possible sum of
    delays of resources in the threshold game, a player's best response must
    first optimize their choice of edges linked to $s$ or $t$. If two players
    choose the edge $(s,\, s_i)$, one of them can improve their latency by
    changing to an unoccupied edge $(s,\,s_{i'})$. Therefore, we can assume the
    $\nth{i}$ player chooses edge $(s,\,s_i)$ WLOG. Player $i$ can also decrease
    their latency by switching from $(t_j,\,t)$ to $(t_{j+1},\,t)$ for any $j < i$
    unless their strategy is limited by the distance bound when $j=i$.

    Player $i$ now has only two strategies from $s_i$ to $t_i$ due to the
    distance bound, corresponding to their strategies in the threshold game:
    \begin{enuminline}
        \item following the right-down path, namely $(s_i,\,v_{i1},\, \cdots,\,
            v_{ii},\, v_{i+1,i},\, \cdots,\, v_{ni}, \, t_i)$, where they occupy
            resources corresponding to $\Sin$
        \item using the edge $(s_i,\,t_i)$, where they occupy the resource
            $\Sout=\{r_i\}$
    \end{enuminline}.
    Clearly PNEs in this \dnc{} correspond to PNEs in the original quadratic
    threshold game.

    Now we determine the values of $w_i$ and $b$. The shortest paths from $s_i$
    to $t_i$ should be either the right-down path or the edge $(s_i,\, t_i)$.
    This implies that $w_i=a_i+b_i+c_i$ where $a_i=i(i-1)+1$ is the total length
    of horizontal edges, $b_i=n+1-i$ is the total length of vertical edges, and
    $c_i=n-1$ is the total length of edges inside $v_{ij}$ for resources
    $r_{ij}$. Hence $w_i=i(i-2) +2n+1$. The bound $b$ should accommodate player
    $n$ who has the longest path and is set as $b=w_n+2=n^2+3$.
\end{proof}

\newcommand{\dbncgBestHardCaption}[1]{
    Illustration of the #1~instance corresponding to a 3-partition problem.
    Double-line edges are slow edges, dashed edges are fast edges, and other
    edges have no delay. Non-unit-length edges have labels to indicate their
    lengths. Deciding whether the total delay can be bounded by $6m-3$ is
    \NP-complete.
}
\begin{theorem}
    Computing the best social welfare (i.e., minimal total delay) among PNEs of
    a \dnc{} is \NP-hard.
    \label{thm:dbncg-best-np-hard}
\end{theorem}
\begin{proof}
    \begin{figure}[tb]
        \centering
        \def\edgeVert{}
        \def\edgeDiag{}
        \def\edgeHori{}
        \resizebox{0.7\textwidth}{!}{%
        \begin{tikzpicture}[
    node distance={7em},
    nodetext/.style = {minimum size=2em, font=\footnotesize},
    main/.style = {draw, nodetext, circle},
    edgelabel/.style = {midway, above, sloped, font=\footnotesize},
    fastedge/.style = {thin, dashed},
    slowedge/.style = {double, double distance=1pt},
    dot/.style = {minimum width=4em, minimum height=2em}
]

    \tikzset{>=Latex}

    \node[main] (t0) {$t_0$};
    \node[main] (t1) [right of=t0] {$t_1$};
    \node[main] (s1) [yshift=-1em][above of=t1] {$s_1$};
    \node[main] (t2) [right of=t1] {$t_2$};
    \node[main] (s2) [yshift=-1em][above of=t2] {$s_2$};
    \node[dot] (tdot) [right of=t2] {};
    \node[dot] (sdot) [yshift=-1em][above of=tdot] {};
    \node[dot] (dispdot) at ($(tdot)!0.5!(sdot)$) {$\cdots$};
    \node[main] (t3m) [right of=tdot] {$t_{3m}$};
    \node[main] (s3m) [yshift=-1em][above of=t3m] {$s_{3m}$};

    \foreach \i in {1,2,3m}
        \draw[->, fastedge] (s\i) to node[edgelabel]{$a_{\i}\edgeVert$} (t\i);

    \foreach \t / \s / \e in {0/1/, 1/2/, 2/dot/.south west, dot/3m/}
        \draw[->] (t\t.15) to node[edgelabel]{\edgeDiag} (s\s\e);

    \foreach \tp/\tn in {0/1, 1/2, 2/dot, dot/3m}
        \draw[->, slowedge] (t\tp) to node[edgelabel]{\edgeHori} (t\tn);
\end{tikzpicture}

% vim: tw=80 filetype=tex foldmethod=marker foldmarker=f{{{,f}}} spell spelllang=en
        }
        \caption{ \dbncgBestHardCaption{\dnc} \label{fig:dbncg-best-hard} }
        \vskip-1em
    \end{figure}

    We reduce from the strongly \NP-complete 3-partition problem~\citep{
    garey1979computers}.

    In the \emph{3-partition problem}, we are given a multiset of $3m$ positive
    integers $S=\condSet{a_i \in \posint}{1 \le i \le 3m}$ and a number $T$ such
    that $\sum a_i = mT$ and $T/4 < a_i < T/2$. The question $Q_1$ is: Can $S$
    be partitioned into $m$ sets $S_1, \cdots, S_m$ such that $\sum_{a_i\in S_j}
    a_i = T$ for all $1 \le j \le m$? Note that due to the strong
    \NP-completeness of 3-partition, we assume the numbers use unary encoding so
    that the \dnc{} graph size is polynomial.

    As in the proof of \cref{thm:dbncg-pls-complete}, we assign a weight
    $w_e\in\posint$ to each edge $e$. The \dnc{} instance has two types of edges
    with non-zero delay: fast edge and slow edge, with delay functions
    $d_{\text{fast}}(x) = \indicator{x\ge1} + 2\indicator{x\ge2}$ and
    $d_{\text{slow}}(x) = 2$.

    As illustrated in \cref{fig:dbncg-best-hard}, for each integer $a_i$, we
    create a pair of vertices $(s_i,\,t_i)$ connected by a fast edge with
    $w_{(s_i,\,t_i)} = a_i$. We create a new vertex $t_0$ as the source while
    using $t_{3m}$ as the sink. For $0 \le i < 3m$, we connect $t_i$ to $t_{i+1}
    $ by a unit-length slow edge and $t_i$ to $s_{i+1}$ by a unit-length edge
    without delay. There are $m$ players who can choose paths with length
    bounded by $b=T+3m$.

    We ask the question $Q_2$: Is there a PNE in the above game where the total
    delay is no more than $m(6m-3)$? Each player prefers an unoccupied fast edge
    to a slow edge but also prefers a slow edge to an occupied fast edge due to
    the above delay functions. Since $T/4 < a_i < T/2$, the best response of a
    player contains either 2 or 3 fast edges, contributing $6m-2$ or $6m-3$ to
    the total delay in either case. Best social welfare of $m(6m-3)$ is only
    achieved when every player chooses 3 fast edges, which also means that their
    choices together constitute a partition of the integer set $S$ in $Q_1$.
    Therefore, $Q_2$ and $Q_1$ have the same answer.
\end{proof}

\paragraph{Remark}
The optimal global welfare of any ``centralized'' solution (where players
cooperate to minimize total delay instead of selfishly minimizing their own
delay) achieves $m(6m-3)$ if and only if the original 3-partition problem has a
solution. Hence we also have the following theorem:
\begin{theorem}
    Computing the optimal global welfare of pure strategies in \dnc{} is
    \NP-hard.
\end{theorem}

\begin{theorem}
    Computing the worst social welfare (i.e., maximal total delay) among PNEs
    of a \dnc{} is \NP-hard.
    \label{thm:dbncg-worst-np-hard}
\end{theorem}
\begin{proof}
    We build on the proof of \cref{thm:dbncg-best-np-hard}. We create a new
    vertex $s$ as the source and connect $s$ to $t_0$ and $s_i$ where $1\le i\le
    3m$:
    \begin{align*}
        \begin{array}{rclrclr}
            w_{(s,\,t_0)} &=& 1 &
            d_{(s,\,t_0)}(x) &=& \indicator{x \ge m + 1} \cdot R &
            \multirow{2}{*}{\quad\text{where} R = 9m + 2} \\
            w_{(s,\,s_i)} &=& T + i - a_i + 1 \quad\quad &
            d_{(s,\,s_i)}(x) &=& \indicator{x \ge 2} \cdot R &
        \end{array}
    \end{align*}
    The delay functions on fast and slow edges are changed to $d_{\text{fast}}(x)
    =  2\indicator{x\ge2} + 2\indicator{x\ge3}$ and $d_{\text{slow}}(x) = 3$.

    There are $4m$ players in this game with a distance bound $b=T+3m+1$. Since
    $R$ is greater than the delay on any path from $s_i$ or $t_0$ to the sink,
    we can assume WLOG that player $i$ choose $(s,\,s_i)$ where $1 \le i \le 3m$,
    and players $3m+1, \cdots, 4m$ all choose $(s,\,t_0)$. The first $3m$
    players generate a total delay of $D_0 = d_{\text{slow}} \cdot 3m(3m-1)/2 =
    9m(3m-1)/2$ where player $i$ occupies one fast edge and $3m-i$ slow edges.
    Each of the last $m$ players occupies 2 or 3 fast edges in their best
    response. Occupying one fast edge incurs $4$ total delay because one of the
    first $3m$ players also uses that edge. Therefore, the each of last $m$
    players contributes $9m+2$ or $9m+3$ to the total delay. We ask the question
    $Q_3$: Is there a PNE where the total delay is at least $D_0 + m(9m+3)$?
    From our analysis, we can see that $Q_3$ and $Q_1$ have the same answer.
\end{proof}
% f}}}

\subsection{Distance-bounded network congestion game with default action}
\label{sec:dbncg-da}
% f{{{

As we have discussed, we formulate capability restriction as limiting the size
of the programs accessible to a player. In this section, we propose a variant of
\dnc{} where we define a DSL to compactly represent the strategies. We will also
show that the size of a program equals the length of the path generated by the
program, which can be much smaller than the number of edges in the path. The new
game, called \emph{distance-bounded network congestion game with default action
(\dbncgda)}, requires that each vertex except the source or sink has exactly one
outgoing zero-length edge as its default action. All other edges have unit
length. A strategy in this game can be compactly described by the actions taken
at \emph{divergent points} where a unit-length edge is followed.

\begin{definition}
    An instance of \dbncgda{} is a tuple $G=(\setV \X \setE \X \setN \X s \X
    t \X b \X (d_e)_{e \in \setE} \X (w_e)_{e\in\setE})$ where:
    \begin{itemize}
        \item $w_e \in \{0,\,1\}$ is the length of edge $e$.
        \item All other symbols have the same meaning as in \cref{def:dbncg}.
    \end{itemize}
    Moreover, we require the following properties:
    \begin{itemize}
        \item A default action, denoted as $\da(\cdot)$, can be defined for
            every non-source, non-sink vertex $v\in\setV/\{s,\, t\}$ such that:
            \begin{align*}
                \Big(v,\,\da(v)\Big) \in \setE,\quad w_{(v,\,\da(v))} &= 0, \\
                \forall u \in \setV/\{\da(v)\}:
                \; (v,\,u)\in\setE &\implies w_{(v,\,u)}=1
            \end{align*}
        \item Edges from the source have unit length: $\forall v \in \setV:\; (s,
            \,v) \in \setE \implies w_{(s,\,v)}=1$
        \item The subgraph of zero-length edges is acyclic. Equivalently,
            starting from any non-source vertex, one can follow the default
            actions to reach the sink.
    \end{itemize}

    The strategy space of a single player contains all $s-t$ simple paths whose
    length does not exceed $b$:
    \begin{align*}
        \strategyLevel{b} \defeq \condSet{(p_0,\, \ldots,\, p_k)}{
            \begin{array}{l}
                p_0=s,\,p_k=t,\,
                (p_i,\,p_{i+1})\in\setE,
                \sum_{i=0}^{k-1}w_{(p_i,\,p_{i+1})} \le b, \\
                p_i \neq p_j \text{ for } i \neq j
            \end{array}}
    \end{align*}
\end{definition}

Note that the strategy spaces are strictly monotonically increasing up to the
longest simple $s-t$ path. This is because for any path $p$ whose length is $b
\ge 2$, we can remove the last non-zero edge on $p$ and follow the default
actions to arrive at $t$, which gives a new path with length $b-1$. Formally, we
have:
\begin{property}
    \label{prop:monostrategy}
    Let $\bmax$ be the length of the longest simple $s-t$ path in a \dbncgda{}
    instance. For $1 \le b < \bmax$, $\strategyLevel{b} \subsetneq
    \strategyLevel{b+1}$
\end{property}

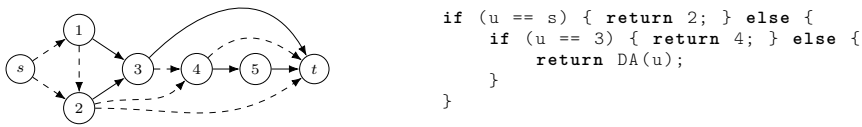
\begin{figure}[tb]
    \subcaptionbox{Example graph structure. Solid arrows are default edges and
        dashed arrows are unit-length edges.}[.53\textwidth]{
        \begin{adjustbox}{width=.36\textwidth}
            \begin{tikzpicture}[
    main/.style = {draw, font=\scriptsize, circle},
    nda/.style = {thin, dashed},
    da/.style = {thin}
]

    \tikzset{>=Latex}

    \node (s) at (0, 0) [main] {$s$};
    \node (1) at (3em, 2em) [main] {$1$};
    \node (2) at (3em, -2em) [main] {$2$};
    \node (3) at (6em, 0) [main] {$3$};
    \node (4) at (9em, 0) [main] {$4$};
    \node (5) at (12em, 0) [main] {$5$};
    \node (t) at (15em, 0) [main] {$t$};

    \draw[->,nda] (s) to (1);
    \draw[->,nda] (s) to (2);
    \draw[->,nda] (1) to (2);
    \draw[->,da] (1) to (3);
    \draw[->,da] (2) to (3);
    \draw[->,da, out=45, in=120, looseness=1.3] (3) to (t);
    \draw[->,da] (4) to (5);
    \draw[->,da] (5) to (t);
    \draw[->,nda, out=15, in=-135, looseness=1.05] (2) to (4);
    \draw[->,nda, out=0, in=-145, looseness=1.1] (2) to (t);
    \draw[->,nda] (3) to (4);
    \draw[->,nda, out=45, in=145, looseness=1.15] (4) to (t);
\end{tikzpicture}

% vim: tw=80 filetype=tex
        \end{adjustbox}
    }
    \hfill
    \subcaptionbox{
        The shortest program to represent the strategy
        $(s,\, 2,\, 3,\, 4,\, 5,\, t)$.
    }[.44\textwidth]{
\begin{lstlisting}[
            language=C, basicstyle=\ttfamily\scriptsize]
if (u == s) { return 2; } else {
    if (u == 3) { return 4; } else {
        return DA(u);
    }
}
\end{lstlisting}
    }
    \caption{
        An example of the \dbncgda{} game and a program to represent a strategy.
        \label{fig:dbncgda-example}
    }
    \vskip-1em
\end{figure}

We define a Domain Specific Language (DSL) with the following context-free
grammar \citep{hopcroft2014introduction} to describe the strategy of a player:
\begin{lstlisting}[xleftmargin=3em,
                   basicstyle=\ttfamily\small,
                   literate=
                   {->}{$\rightarrow$}{2}
                   {in}{$\in$}{2}
                   {setV}{$\setV$}{2}
                   {ret}{{\bf return}}{6}
                   {if}{{\bf if}}{2}
                   {else}{{\bf else}}{4}
                   {vu}{$u$}{1}
                   {vv}{$v$}{1}
               ]
Program -> ret DA(vu);
           | if (vu == V) {ret V;} else {Program}
      V -> vv in setV
\end{lstlisting}

A program $p$ in this DSL defines a computable function $f_p: \setV \mapsto
\setV$ with semantics similar to the C language where the input vertex is stored
in the variable $u$, as illustrated in \cref{fig:dbncgda-example}. The strategy
corresponding to the program $p$ is a path $(c_0,\, \dots,\,c_k)$ from $s$ to
$t$ where:
\begin{align*}
    c_0 = s \quad\quad
    c_{i+1} = f_p(c_i) \text{ for } i \ge 0 \text{ and } c_i \neq t \quad\quad
    k = i \text{ if } c_i=t
\end{align*}

We define the capability of a player as the maximum size of programs that they
can use. The size of a program is the depth of its parse tree. Due to the
properties of \dbncgda{}, the shortest program that encodes a path from $s$ to
$t$ specifies the edge chosen at all divergent points in this path. The size of
this program equals the length of the path. Hence the distance bound in the game
configuration specifies the capability of each player in the game. To study the
game outcome under different player capability constraints, we study \dbncgda{}
instances with different values of $b$.

We state hardness results for \dbncgda{}. Their proofs are similar to those
for \dnc{} except that we need to redesign the edge weights to conform to the
requirements on default action. The proof details are given in
\cref{append:dbncg-da-proof}.

\begin{restatable}{theorem}{ThmDbncgDaPlsComplete}
    \label{thm:dbncg-da-pls-complete}
    \dbncgda{} is \PLS-complete.
\end{restatable}

\begin{restatable}{theorem}{ThmDbncgDaBestNpHard}
    \label{thm:dbncg-da-best-np-hard}
    Computing the best social welfare (i.e., minimal total delay) among PNEs
    of a \dbncgda{} is \NP-hard.
\end{restatable}

\begin{theorem}
    \label{thm:dbncg-da-worst-np-hard}
    Computing the worst social welfare (i.e., maximal total delay) among PNEs
    of a \dbncgda{} is \NP-hard.
\end{theorem}

\begin{theorem}
    Computing the optimal global welfare of pure strategies in \dbncgda{} is
    \NP-hard.
\end{theorem}

% f}}}

% vim: tw=80 filetype=tex foldmethod=marker foldmarker=f{{{,f}}} spell spelllang=en

\section{Impact of player capability on social welfare in {\mdseries\dbncgda}}
\label{sec:trend-dbncgpa}

We first introduce four \wltrend{} properties for general games. Given a game
with a finite hierarchy of player capabilities, we use $\strategyLevel{b}$ to
denote the strategy space when player capability is bounded by $b$. Assuming the
maximal capability is $\bmax$ (which is the longest $s-t$ simple path in
\dbncgda{}), we have $\strategyLevel{b} \subsetneq \strategyLevel{b+1}$ for
$1\le b < \bmax$ (see \cref{prop:monostrategy}). We use $\equil(b) \subseteq
\strategyLevel{b}^n$ to denote the set of all PNEs at the capability level $b$.
We define $\bestw{b} \defeq \max_{\vs\in\equil(b)} W(\vs)$ to be the best social
welfare at equilibrium and $\worstw{b} \defeq \min_{\vs\in\equil(b)} W(\vs)$ the
worst social welfare.

\begin{definition}
    A game is \pp{} if social welfare at equilibrium cannot decrease as players
    become more capable, i.e.,
    $\forall 1 \le b < \bmax$, $\bestw{b} \leq \worstw{b+1}$.
\end{definition}

\begin{definition}
    A game is \bwr{} if the worst social welfare at equilibrium under maximal
    player capability is at least as good as any social welfare at equilibrium
    under lower player capability, i.e.,
    $\forall 1 \leq b < \bmax$, $\bestw{b} \leq \worstw{\bmax} $.
\end{definition}

Note that \bwr{} is a weaker condition than \pp{}. We then define analogous
properties for games where less capable players lead to better outcomes:

\begin{definition}
    A game is \ap{} if social welfare at equilibrium cannot increase as players become more capable, i.e.,
    $\forall 1 \leq b < \bmax $, $\bestw{b+1} \leq \worstw{b}$.
\end{definition}

\begin{definition}
    A game is \bfr{} if the worst social welfare at equilibrium under minimal player capability is at least as good as any social welfare at equilibrium under higher player capability, i.e.,
    $\forall b \geq 2$, $\bestw{b} \leq \worstw{1}$.
\end{definition}

Our goal is to identify games that guarantee these properties.
Since solving equilibria for general \dbncgda{} is computationally hard
(\cref{sec:dbncg-da}), we focus on a restricted version of \dbncgda{} where all
edges share the same delay function; formally, we consider the case
$\forall e\in \setE: d_e(\cdot) = \fd$ where $\fd$ is non-negative and
non-decreasing. We call this game \emph{distance-bounded network congestion game
with default action and shared delay (\dbncgdas{})}. We aim to find conditions
on $\fd$ under which the properties hold \emph{universally} (i.e., for all
network configurations of \dbncgdas{}). \Cref{tab:property} summarizes the
results.

\begin{table}[tb]
    \setlength{\tabcolsep}{2pt}
    \def\arraystretch{1.2}
    \footnotesize
    \centering
    \begin{tabular}{>{\centering\arraybackslash}m{0.26\linewidth} m{0.35\linewidth} m{0.35\linewidth}}
    \toprule
    & \dbncgdas{} (\cref{sec:trend-dbncgpa}) & \gam{} (\cref{sec:gm,sec:trend-gm})\\
    \midrule
    Resource layout & On a directed graph & On parallel horizontal lines \\
    Strategy space & Paths from $s$ to $t$ & Piecewise-constant functions \\
    Delay (payoff) & Non-negative non-decreasing & $\frg$ positive, $\frm$ negative \\
    \midrule
    \pp{} & $\fd$ is a constant function & $\frg, \frm$ are constant functions\\
    \bwr{} & $\fd$ is a constant function & $w(x) = xr_g(x)$ attains maximum at $x=n$ \\
    \ap{} & $\fd$ is the zero function & Never \\
    \bfr{} & $\fd$ is the zero function & Never \\
    \bottomrule
    \end{tabular}
    \vskip.5em
    \caption{Necessary and sufficient conditions on the delay or payoff
        functions such that the \wltrend{} properties hold universally.
        \label{tab:property}
    }
    \vskip-2em
\end{table}

\begin{restatable}{theorem}{dbpp}
\label{thm:db-pp}
    \dbncgdas{} is universally \pp{} if and only if $\fd$ is a constant
    function.
\end{restatable}

\begin{proof}
    If $\fd$ is a constant function, the total delay achieved by a strategy is
    not affected by the load condition of each edge (thus not affected by other
    players' strategies). So each player's strategy in any PNE is the one in
    $\strategyLevel{b}$ that minimizes the total delay under the game layout.
    Denote this minimum delay as $\delta(b)$. For any $b \ge 1$, we have
    $\strategyLevel{b} \subseteq \strategyLevel{b+1}$, so $\delta(b) \geq
    \delta(b+1)$. And for any $\vs \in \equil(b)$, $W(\vs) = - n \delta(b)$.
    Hence $\bestw{b} \leq \worstw{b+1}$.

    If $\fd$ is not a constant function, we show that there exists an instance
    of \dbncgdas{} with delay function $\fd$ that is not \pp{}. Define $v =
    \min\condSet{x}{d(x) \neq d(x+1)}$. It follows that $d(v') = d(v)$ for all
    $v'\leq v$. We consider the cases $d(v) = 0$ and $d(v) > 0$ separately.

    \vspace{1em}

    \noindent\textbf{Case 1: $d(v) > 0$}\hspace{1em}
    Denote $\rho = \frac{d(v+1)}{d(v)}$. Since $\fd$ is non-decreasing, $\rho >
    1$. We construct a game with the network layout in \cref{fig:da-pp-pos}
    with $n=v+1$ players. We will show that $\bestw{1} > \worstw{2}$.

    First, it is easy to see that the PNEs when $b=1$ are $a$ players take the
    upper path and $v+1-a$ players take the lower path, where $1\leq a\leq v$.
    All PNEs achieves a social welfare of $W_1 = - (v+1) (N_1+N_2+3)d(v)$.

    We set $N_1 = \lfloor \frac{1}{\rho - 1}\rfloor$ and $N_2 = \lfloor (N_1+2)
    \rho\rfloor - 1$. When $b=2$, one PNE is that all players choose the path
    from upper left to lower right using the crossing edge in the middle due to
    our choice of $N_1$ and $N_2$. Its social welfare $W_2 = - (v+1) (2N_1 + 3)
    d(v+1)$. One can check that $W_1 > W_2$, hence $\bestw{1} > \worstw{2}$.
    \Cref{append:db-pp} presents more details.

    \vspace{1em}

    \noindent\textbf{Case 2: $d(v) = 0$}\hspace{1em}
    We construct a game with the network layout in \cref{fig:da-pp-zero} where
    there are $2v$ players. With $b=1$, half of the players choose
    the upper path and the others choose the lower path, which has a social
    welfare $W_1=0$. With $b=2$, a PNE is:
    \begin{enuminline}
        \item $v$ players take the path $(s,\, N_1 \text{ edges},\, \text{lower
            right } N_2 \text{ edges},\, t)$
        \item the other $v$ players take the path $(s,\, \text{lower left } N_2
            \text{ edges},\, N_1 \text{ edges},\, t)$
    \end{enuminline}.
    We choose $N_1$ and $N_2$ to be positive integers that satisfy $\frac{N_2+1}
    {N_1} > \frac{d(2v)}{d(v+1)}$. The social welfare
    $W_2=-2vN_1d(2v) < W_1$.
    Hence $\bestw{1} > \worstw{2}$.
\end{proof}

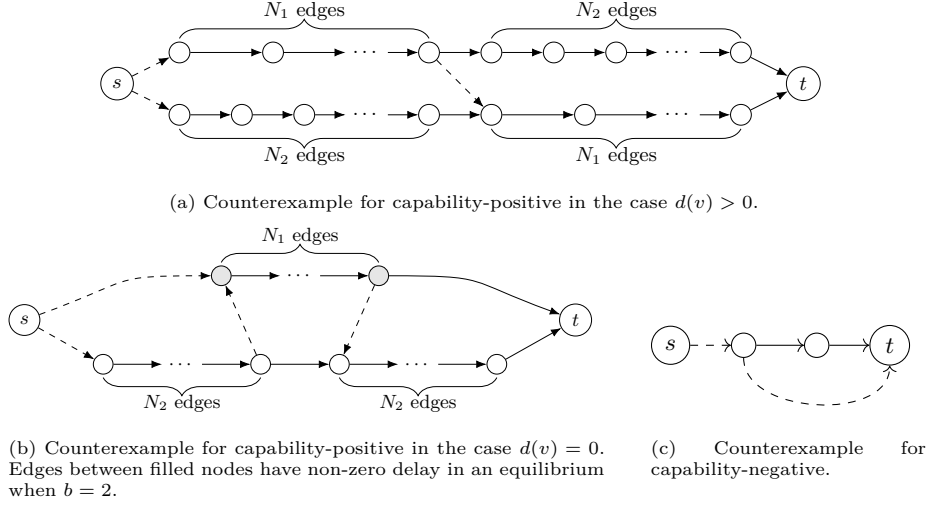
\begin{figure}[tb]
    \centering
    \begin{subfigure}[t]{\textwidth}
        \centering
        \resizebox{0.8\textwidth}{!}{%
            \begin{tikzpicture}[
    node distance={3em},
    nodetext/.style = {minimum size=2em, font=\small},
    main/.style = {draw, nodetext, circle, minimum size=1em},
    edgelabel/.style = {midway, above, sloped, font=\tiny},
    crossedge/.style = {thin, dashed},
    dot/.style = {minimum width=2em, minimum height=2em}
]

    \tikzset{>=Latex}
    
    \node[main] (d1) {};
    \node[main] (d2) [right of=d1] {};
    \node[main] (d3) [right of=d2] {};
    \node[dot] (d4) [right of=d3] {$\cdots$};
    \node[main] (d5) [right of=d4] {};
    
    \node[main] (u1) [above of=d1] {};
    \node[dot] (ud1) [above of=d2] {};
    \node[dot] (ud2) [above of=d3] {};
    \node[main] (u2) at ($(ud1)!0.5!(ud2)$) {};
    \node[dot] (u3) [above of=d4] {$\cdots$};
    \node[main] (u4) [above of=d5] {};
    
    \foreach \i in {5,...,9} {
        \pgfmathtruncatemacro{\j}{\i - 1}
        \ifthenelse{\i = 8}
        {\node[dot] (u\i) [right of=u\j] {$\cdots$}}
        {\node[main] (u\i) [right of=u\j] {}};
    }
    
    \node[main] (d6) [below of=u5] {};
    \node[dot] (dd1) [below of=u6] {};
    \node[dot] (dd2) [below of=u7] {};
    \node[main] (d7) at ($(dd1)!0.5!(dd2)$) {};
    \node[dot] (d8) [below of=u8] {$\cdots$};
    \node[main] (d9) [below of=u9] {};
    
    \foreach \i in {1,...,8} {
        \pgfmathtruncatemacro{\j}{\i + 1}
        \draw[->] (d\i) to (d\j);
        \draw[->] (u\i) to (u\j);
    }
    
    \draw [decorate,decoration={brace,amplitude=3mm,raise=.5mm}] (u1.north) -- (u4.north);
    \node[dot] (ut1) at ([yshift=2em]$(u1)!0.5!(u4)$) {$N_1$ edges};
    \draw [decorate,decoration={brace,amplitude=3mm,raise=.5mm}] (u5.north) -- (u9.north);
    \node[dot] (ut1) at ([yshift=2em]$(u5)!0.5!(u9)$) {$N_2$ edges};
    \draw [decorate,decoration={brace,mirror,amplitude=3mm,raise=.5mm}] (d1.south) -- (d5.south);
    \node[dot] (dt1) at ([yshift=-2em]$(d1)!0.5!(d5)$) {$N_2$ edges};
    \draw [decorate,decoration={brace,mirror,amplitude=3mm,raise=.5mm}] (d6.south) -- (d9.south);
    \node[dot] (dt1) at ([yshift=-2em]$(d6)!0.5!(d9)$) {$N_1$ edges};
    
    \node[main] (s) at ([xshift=-3em]$(u1)!0.5!(d1)$) {$s$};
    \node[main] (t) at ([xshift=3em]$(u9)!0.5!(d9)$) {$t$};
    
    \foreach \i/\j in {s/u1, s/d1, u4/d6}
        \draw[->, crossedge] (\i) to (\j);
    \foreach \i/\j in {u9/t, d9/t}
        \draw[->] (\i) to (\j);

\end{tikzpicture}
        }
        \caption{Counterexample for \pp{} in the case $d(v) > 0$.
        \label{fig:da-pp-pos}}
    \end{subfigure}
    \begin{subfigure}[t]{0.65\textwidth}
        \centering
        \resizebox{\textwidth}{!}{%
            \begin{tikzpicture}[
    node distance={4em},
    nodetext/.style = {minimum size=2em, font=\small},
    main/.style = {draw, nodetext, circle, minimum size=1em},
    edgelabel/.style = {midway, above, sloped, font=\tiny},
    crossedge/.style = {thin, dashed},
    dot/.style = {minimum width=2em, minimum height=2em}
]

    \definecolor{marknode}{rgb}{0.9,0.9,0.9}

    \tikzset{>=Latex}

    \node[main] (d1) {};
    \foreach \i in {2,...,6} {
        \pgfmathtruncatemacro{\j}{\i - 1}
        \ifthenelse{\i = 2 \OR \i = 5}
        {\node[dot] (d\i) [right of=d\j] {$\cdots$}}
        {\node[main] (d\i) [right of=d\j] {}};
        \draw[->] (d\j) -- (d\i);
    }

    \foreach \i in {1,4} {
        \pgfmathtruncatemacro{\j}{\i + 2}
        \draw[decorate, decoration={brace,mirror,amplitude=3mm,raise=.5mm}]
            (d\i.south) -- (d\j.south);
        \node[dot] at ([yshift=-2em]$(d\i)!0.5!(d\j)$) {$N_2$ edges};
    }

    \node[dot] (u2) at ([yshift=4.5em]$(d3)!0.5!(d4)$) {$\cdots$};
    \node[main, fill=marknode] (u1) [left of=u2] {};
    \node[main, fill=marknode] (u3) [right of=u2] {};
    \draw[->] (u1) -- (u2);
    \draw[->] (u2) -- (u3);
    \draw[->, crossedge] (u3) -- (d4);
    \draw[->, crossedge] (d3) -- (u1);
    \draw[decorate, decoration={brace,amplitude=3mm,raise=.5mm}]
        (u1.north) -- (u3.north);
    \node[dot] at ([yshift=2em]$(u1)!0.5!(u3)$) {$N_1$ edges};

    \node[main] (s) at ([xshift=-7em]$(u1)!0.5!(d1)$) {$s$};
    \node[main] (t) at ([xshift=7em]$(u3)!0.5!(d6)$) {$t$};

    \draw[->, crossedge] (s) -- (d1);
    \draw[->, crossedge, out=20, in=-180, looseness=1.5] (s) to (u1);
    \draw[->] (d6) -- (t);
    \draw[->, out=0, in=160, looseness=1.5] (u3) to (t);
\end{tikzpicture}

% vim: tw=80 filetype=tex
        }
        \caption{Counterexample for \pp{} in the case $d(v) = 0$. Edges between filled nodes have non-zero delay in an equilibrium when $b=2$.\label{fig:da-pp-zero}}
    \end{subfigure}\hfill
    \begin{subfigure}[t]{0.3\textwidth}
        \centering
        \resizebox{\textwidth}{!}{%
        \begin{tikzpicture}[
    node distance={3em},
    nodetext/.style = {minimum size=2em, font=\small},
    main/.style = {draw, nodetext, circle, minimum size=1em},
    edgelabel/.style = {midway, above, sloped, font=\tiny},
    crossedge/.style = {thin, dashed},
    dot/.style = {minimum width=2em, minimum height=2em}
]

    \node[main] (s) {$s$};
    \node[main] (1) [right of=s] {};
    \node[main] (2) [right of=1] {};
    \node[main] (t) [right of=2] {$t$};
    
    \draw[->] (1) to (2);
    \draw[->] (2) to (t);
    \draw[->, crossedge] (s) to (1);
    \draw[->, crossedge, out=-90, in=-90] (1.south) to (t.south);
    
\end{tikzpicture}
        }
        \caption{Counterexample for \ap{}.\label{fig:da-ap}}
    \end{subfigure}
    \caption{
        Counterexamples when $\fd$ does not meet the conditions. Dashed arrows
        denote unit-length edges and solid arrows denote zero-length edges
        (default action). Every edge shares the same delay function $\fd$.
    }
    \vskip-1em
\end{figure}

\paragraph{Remark}
The proof for the sufficient condition also holds when different edges have
different delay functions. So the following statement is also true: \dbncgda{}
is universally \pp{} if all edges have constant delay functions.

\begin{theorem}
    \label{thm:db-bwr}
    \dbncgdas{} is universally \bwr{} if and only if $\fd$ is a constant
    function.
\end{theorem}
\begin{proof}
    The ``if'' part follows from \cref{thm:db-pp} since a \pp{} game is also
    \bwr{}. The constructed games in the proof of \cref{thm:db-pp} also serve as
    the counterexamples to prove the ``only if'' part.
\end{proof}

\begin{theorem}
    \label{thm:db-ap}
    \dbncgdas{} is universally \ap{} if and only if $\fd$ is the zero function.
\end{theorem}
\begin{proof}
    If $\fd=0$, then all PNEs have welfare 0, which implies \ap.
    If $\fd$ is not the zero function, denote $v = \min \condSet{x}{d(x) \neq 0}$.
    We construct a game with the network layout shown in \cref{fig:da-ap} with
    $n=v$ players. When $b=1$, all players use the only strategy with a social
    welfare $W_1 = - 3vd(v)$. When $b=2$: if $v=1$, the player will choose both
    dashed paths and achieves $W_2 = -2d(1)$; if $v\geq 2$, the players will
    only experience delay on the first edge by splitting between the default
    path and the shortcut dashed path, which achieves a welfare $W_2 = -vd(v)$.
    In both cases, the game is not \ap{} since $\worstw{1} \le W_1 < W_2 \le
    \bestw{2}$.
\end{proof}

The same argument can also be used to prove the following result:
\begin{theorem}
    \label{thm:db-bfr}
    \dbncgdas{} is universally \bfr{} if and only if $\fd$ is the zero function.
\end{theorem}

% vim: tw=80 filetype=tex foldmethod=marker foldmarker=f{{{,f}}} spell spelllang=en

\section{Gold and Mines Game}
\label{sec:gm}
% f{{{

In this section, we introduce a particular form of \dbncgda{} called Gold and Mines Game (\gam{}). It provides a new perspective on how to define the strategy space hierarchy in congestion games. It also enables us to obtain additional characterizations of how social welfare at equilibrium varies with player capability. Intuitively, as shown in \cref{fig:aog}, a \gam{} instance consists of a few parallel horizontal lines and two types of resources: gold and mine. Resources are placed at distinct horizontal locations on the lines. A player's strategy is a piecewise-constant function to cover a subset of resources. The function is specified by a program using if-statements.

\begin{definition}
An instance of \gam{} is a tuple $G = (\setE, K, \setN, r_g, r_m, b)$ where:
\begin{itemize}
    \item $\setE$ is the set of resources. Each resource $e\in\setE$ is described by a tuple $(x_e, y_e, \alpha_e)$, where $(x_e, y_e)$ denotes the position of the resource in the $x$-$y$ plane, and $\alpha_e \in \{\textrm{gold, mine}\}$ denotes the type of the resource. Each resource has a distinct value of $x$, i.e. $x_e \neq x_{e'}$ for all $e\neq e'$.
    \item $K\in \nat$ is the number of lines the resources can reside on. All resources are located on lines $y=0$, $y=1$, ..., $y=K-1$, i.e. $\forall e, y_e \in \{0,1,\dots,K-1\}$.
    \item $\setN=\{1,\, \dots,\, n\}$ denotes the set of players.
    \item $r_g: \nat \mapsto \real^+$ is the payoff function for gold. $r_g$ is a positive function.
    \item $r_m: \nat \mapsto \real^-$ is the payoff function for mine. $r_m$ is a negative function.
    \item $b\in \nat$ is the level in the strategy space hierarchy defined by the domain-specific language $\mathcal{L}$ (defined below). The strategy space is then $\strategyLevel{b}$.
\end{itemize}
\end{definition}

The strategy $s_i$ of player $i$ is represented by a function $f_i(\cdot)$ that conforms to a domain-specific language $\mathcal{L}$ with the following grammar:
\begin{lstlisting}[xleftmargin=1em,
                   basicstyle=\ttfamily,
                   basicstyle=\footnotesize,
                   literate=
                   {->}{$\rightarrow$}{2}
                   {in}{$\in$}{2}
                   {setR}{$\real$}{2}
                   {ret}{{\bf return}}{6}
                   {if}{{\bf if}}{2}
                   {else}{{\bf else}}{4}
                   {vx}{$x$}{1}
                   {vt}{$t$}{1}
                   {vk}{$K$}{1}
                   {dot}{$\dots$}{3}
               ]
Program -> ret C; | if (vx < vt) {ret C;} else {Program}
C       -> 0 | 1 | dot | vk - 1
t       in setR
\end{lstlisting}
This DSL defines a natural strategy space hierarchy by restricting the number of if-statements in the program. A program with $b-1$ if-statements represents a piecewise-constant function with at most $b$ segments. We denote $\strategyLevel{b}$ as the level $b$ strategy space which includes functions with at most $b-1$ if-statements.

Player $i$ \emph{covers} the resources that their function $f_i$ passes: $E_i = \condSet{e}{f_i(x_e) = y_e}$. The load on each resource is the number of players that covers it: $x_e =  |\condSet{i}{e \in E_i}|$. Each player's payoff is $u_i = \sum_{e\in E_i} r_e(x_e)$, where $r_e$ is either $r_g$ or $r_m$ depending on the resource type. The social welfare is $W(\vs) = \sum_{i\in\setN} u_i$.

Note that \gam{} can be represented as a \dbncgda{} (\cref{append:repr}).
As a result, the problem of computing a PNE in a \gam{} belongs to \PLS.

% f}}}

\section{Impact of player capability on social welfare in \gam{}}
\label{sec:trend-gm}
% f{{{

We revisit the question of how player capability affects the social welfare at equilibrium in the context of \gam{}. We consider the same \wltrend{} properties as in \cref{sec:trend-dbncgpa}. The results are summarized in \cref{tab:property}.

\begin{restatable}{theorem}{gmpp}
\label{thm:gm-pp}
\gam{} is universally \pp{} if and only if both $\frg$ and $\frm$ are constant functions.
\end{restatable}
The idea of the proof is similar to that of \cref{thm:db-pp}. The construction of the counterexamples for the ``only if'' part requires some careful design on the resource layout. \Cref{append:gm-pp} presents the proof.

\begin{theorem}
\label{thm:gm-bwr}
Define welfare function for gold as $w_g(x) \defeq x\cdot r_g(x)$. \gam{} is universally \bwr{} if and only if $w_g$ attains its maximum at $x=n$ (i.e. $\max_{x\le n} w_g(x) = w_g(n)$), where $n$ is the number of players.
\end{theorem}

\begin{proof}
We first notice that there is only one PNE when $b = \bmax$, which is all players cover all gold and no mines. This is because $r_g$ is a positive function and $r_m$ is a negative function, and since all $x_e$'s are distinct, each player can cover an arbitrary subset of the resources when $b = \bmax$. So $W_{\bmax} = nM_g r_g(n)$ where $M_g$ is the number of gold in the game.

If $\max_{x \le n} w_g(x) = w_g(n)$, we show that $W_{\bmax}$ is actually the maximum social welfare over all possible strategy profiles of the game. For any strategy profile $\vs$ of the game, the social welfare
\begin{align*}
    W(\vs) = \sum_{i\in\setN} \sum_{e\in E_i} r_e(x_e) = \sum_{e\in\setE} x_e\cdot r_e(x_e) \leq \sum_{e\in\setE_g} n\cdot r_g(n) = nM_g r_g(n) = W_{\bmax}.
\end{align*}
Therefore, the game is \bwr{}.

If $\max_{x \le n} w_g(x) > w_g(n)$, denote $n' = \argmax_{x\le n} w_g(x)$, $n' < n$. We construct a game with the corresponding $\frg$ that is not \bwr{}.
The game has $n$ players, $K=n$ lines, and each line has one gold. The only PNE in $\equil(\bmax)$ is all players cover all gold, which achieves a social welfare of $W_{\bmax} = n \cdot w_g(n)$. When $b=n'$, one PNE is each player covers $n'$ gold, with each gold covered by exactly $n'$ players. This can be achieved by letting player $i$ cover the gold on lines $\{y=(j\mod n)\}_{j=i}^{i+n'-1}$. To see why this is a PNE, notice that any player's alternative strategy only allows them to switch to gold with load larger than $n'$. For all $x> n'$, since $w_g(n')\geq w_g(x)$, $r_g(n') > r_g(x)$. So such change of strategy can only decrease the payoff of the player. The above PNE achieves a social welfare $W_{n'} = n\cdot w_g(n') > W_{\bmax}$, so the game is not \bwr{}.

\end{proof}

\begin{restatable}{theorem}{gmbfr}
\label{thm:gm-bfr}
For any payoff functions $\frg$ and $\frm$, there exists an instance of \gam{} where \bfr{} does not hold (therefore \ap{} does not hold either).
\end{restatable}

\begin{proof}
It is trivial to construct such a game with mines. Here we show that for arbitrary $\frg$, we can actually construct a game with only gold that is not \bfr{}.

Let $r_{\min} = \min_{x\leq n} r_g(x)$ and $r_{\max} = \max_{x\leq n} r_g(x)$. We construct a game with $K=2$ lines and $N+1$ gold where $N > \frac{r_{\max}}{r_{\min}}$. In the order of increasing $x$, the first $N$ gold is on $y=0$ and the final gold is on $y=1$. When $b=1$, for an arbitrary player, denoting the payoff of choosing $y=0$ (resp. $y=1$) as $r_0$ (resp. $r_1$). Then $r_0 = \sum_{e\in\setE_0} r_g(x_e) \geq \sum_{e\in\setE_0} r_{\min} = Nr_{\min} > r_{\max} \geq r_1$, where $\setE_0$ is the set of resources on $y=0$.
So all the players will choose $y=0$ in the PNE. The social welfare is $W_1 = nNr_g(n)$.
When $b=2$, all the players will choose to cover all the gold in the PNE. So the social welfare is $W_2 = n(N+1)r_g(n) > W_1$. Therefore, the game is not \bfr{}.

\end{proof}

% f}}}

% vim: tw=0 filetype=tex foldmethod=marker foldmarker=f{{{,f}}} spell spelllang=en

\section{Case study: \aog{}}

In this section, we present a special form of \gam{} called the \aog{}. We derive exact expressions of the social welfare at equilibrium with respect to the capability bound. The analysis provides insights on the factors that affect the trend of social welfare over player capability.

\begin{definition}
The \aog{} is a special form of the \gam{}, with $n=2$ players and $K=2$ lines. The layout of the resources follows an alternating ordering of gold and mines as shown in \cref{fig:aog}. Each line has $M$ mines and $M+1$ gold. The payoff functions satisfy $0 < r_g(2) < \frac{r_g(1)}{2}$ (reflecting competition when both players occupy the same gold) and $r_m(1) = r_m(2) < 0$. WLOG, we consider normalized payoff where $r_g(1) = 1, r_g(2) = \rho, 0< \rho < \frac{1}{2}, r_m(1) = r_m(2) = \mu < 0$.

\end{definition}

\begin{figure}[tb]
    \centering
    \resizebox{0.6\textwidth}{!}{%
    \begin{tikzpicture}[
    dot/.style = {minimum width=2em, minimum height=2em}
]

\definecolor{linecolor}{rgb}{0.8,0.8,0.8}
\definecolor{funccolorA}{rgb}{1.0, 0.0, 0.0} % red
\definecolor{funccolorB}{rgb}{0.0, 0, 1.0} % blue

\draw[linecolor, ultra thin] (-1, 0) -- (8.5, 0);
\draw[linecolor, ultra thin] (-1, 0.7) -- (8.5, 0.7);

\foreach \i in {0,1,2.5} {
    \filldraw[black] (\i*2,0.7) circle (2pt) node{};
    \filldraw[black] (\i*2+0.5,0) circle (2pt) node{};
    \draw (\i*2+1,0.7) node[cross=2pt] {};
    \draw (\i*2+1.5,0) node[cross=2pt] {};
}

\node[dot] at (4.2, 0.35) {$\dots$};

\filldraw[black] (3.5*2,0.7) circle (2pt) node{};
\filldraw[black] (3.5*2+0.5,0) circle (2pt) node{};

\draw[funccolorA, dashed, thick] (-0.5, 0.7) -- (2.25, 0.7) --
            (2.25, 0.0) -- (8, 0.0);
\draw[funccolorB, dashed, thick] (-0.5, -0.0) -- (3.25, -0.0) --
            (3.25, 0.7) -- (8, 0.7);

\end{tikzpicture}
    }
    \caption{
        Resource layout for the \aog{}.
        Each dot (resp. cross) is a gold (resp. mine). The dashed lines represent a PNE when $b=2$ (with $-2+\rho < \mu < -\rho$).
        \label{fig:aog}
    }
    \vskip-.5em
\end{figure}
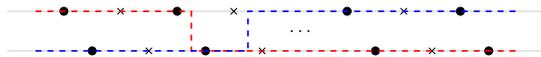

Let's consider the cases $b=1$ and $b=2$ to build some intuitive understanding.
When $b=1$, the PNE is that each player covers one line, which has social welfare $W_1 = 2M+2M\mu+2$. When $b=2$ (and $-2+\rho < \mu < -\rho$), one PNE is shown in \cref{fig:aog}, where the players avoid one mine but cover one gold together, which has social welfare $W_2 = W_1-1-\mu+2\rho$. Whether the social welfare at $b=2$ is better depends on the sign of $2\rho-\mu-1$. In fact, we have the following general result:

\begin{figure}[hbt]
    \centering
    \begin{subfigure}[]{0.32\textwidth}
    \centering
    \includegraphics[width=\textwidth]{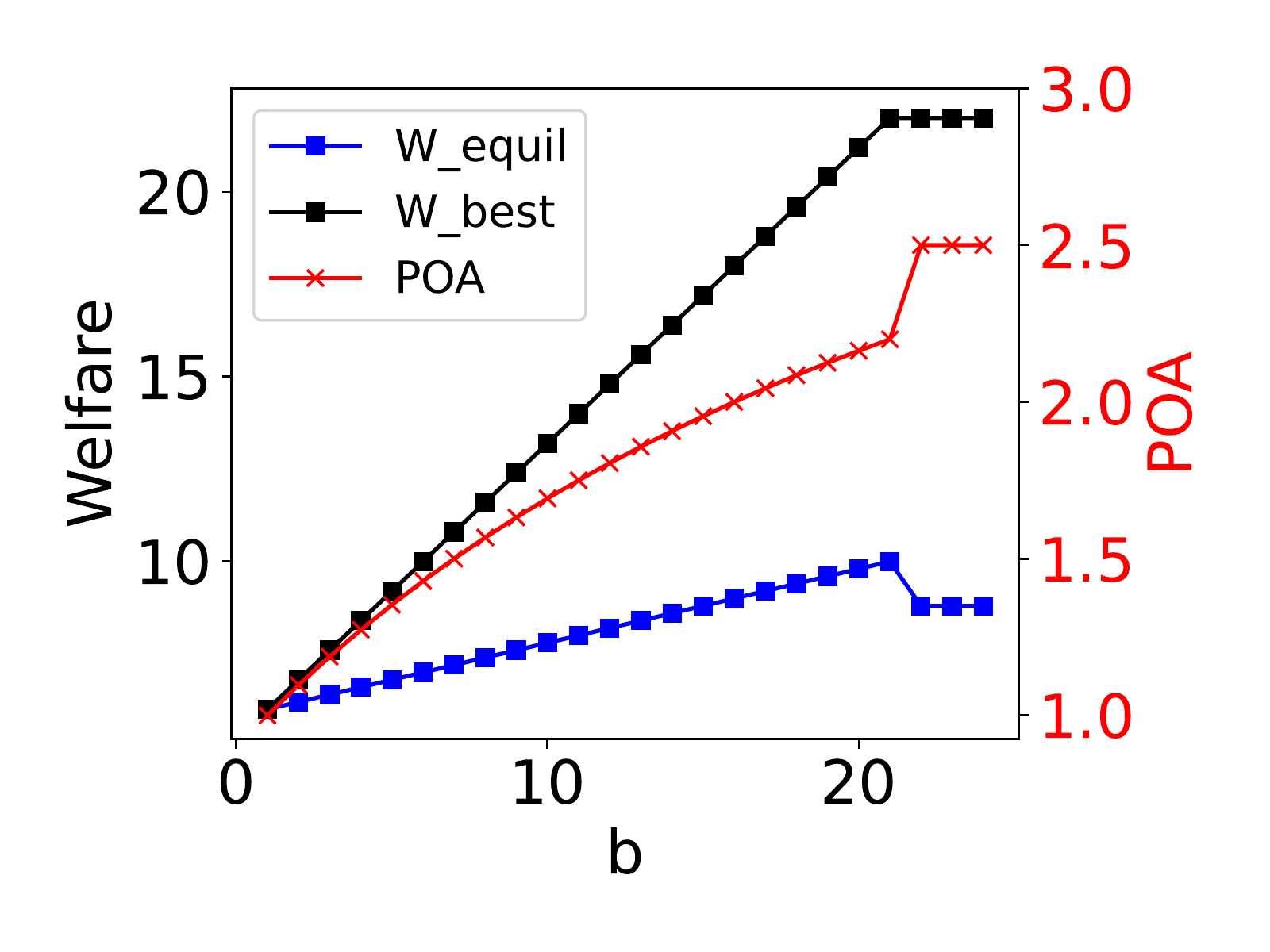}
    \caption{$\mu = -0.8$ ($\mu < 2\rho - 1$)}
    \end{subfigure}
    \begin{subfigure}[]{0.32\textwidth}
    \centering
    \includegraphics[width=\textwidth]{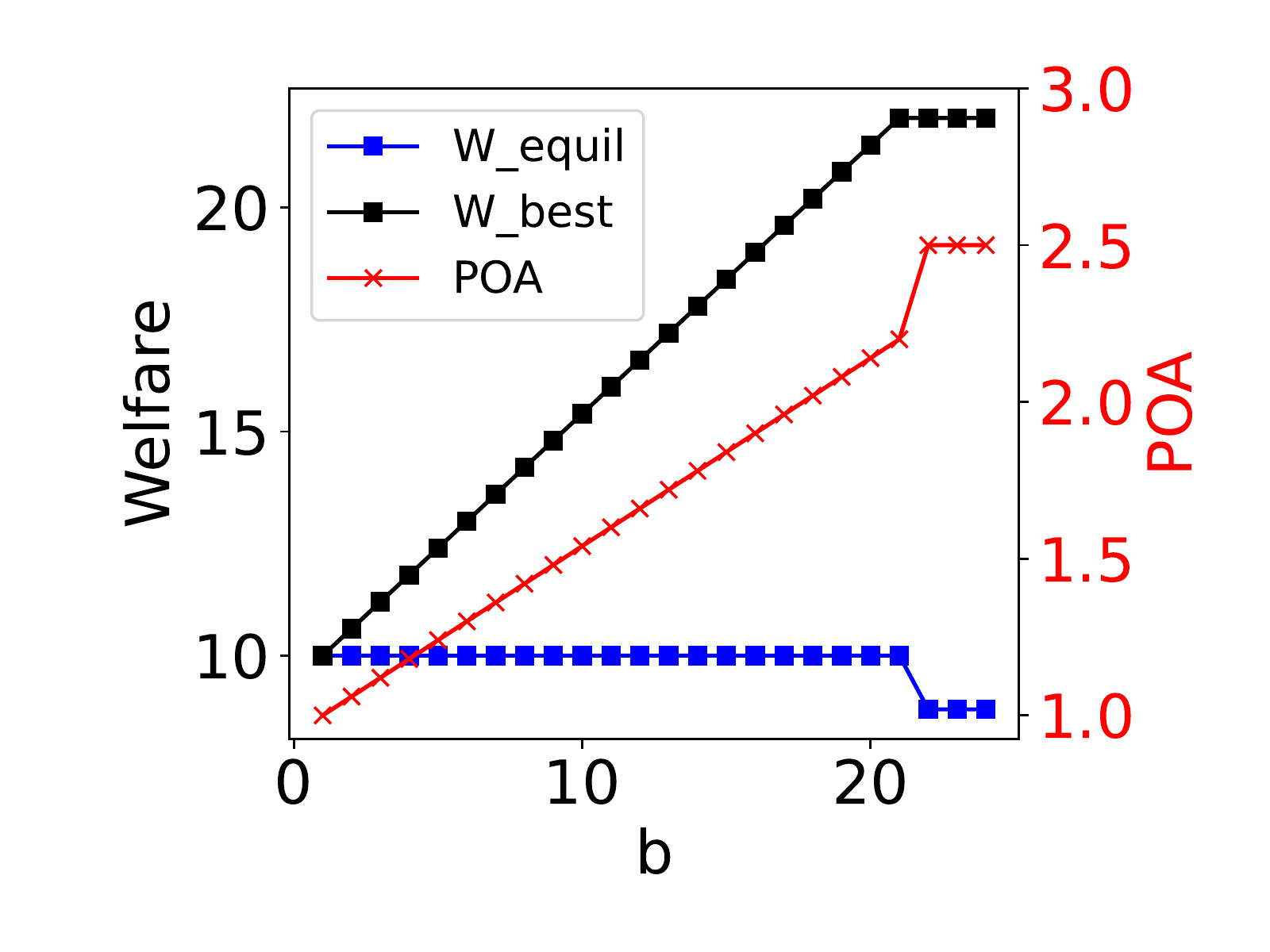}
    \caption{$\mu = -0.6$ ($\mu = 2\rho - 1$)}
    \end{subfigure}
    \begin{subfigure}[]{0.32\textwidth}
    \centering
    \includegraphics[width=\textwidth]{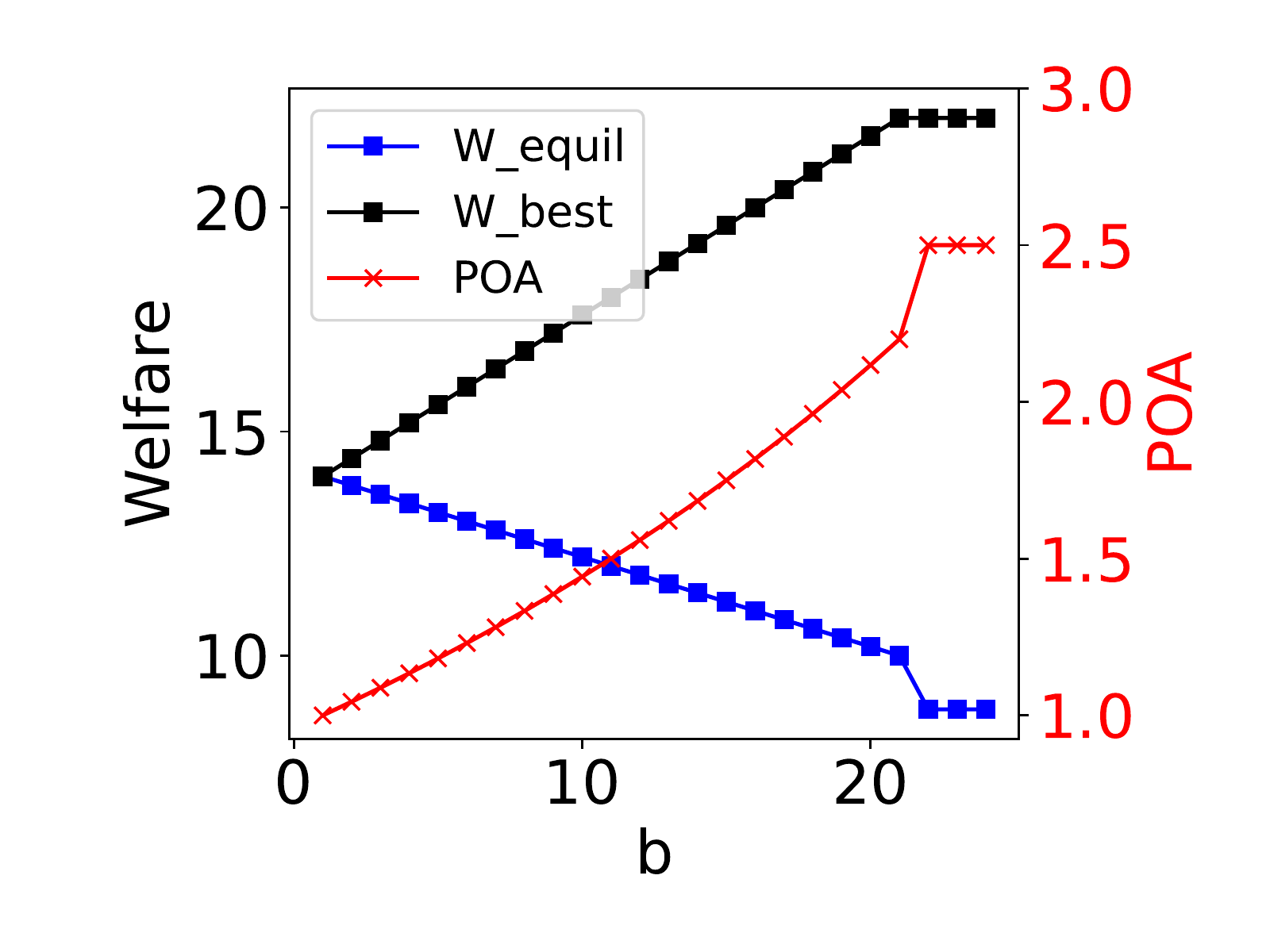}
    \caption{$\mu = -0.4$ ($\mu > 2\rho - 1$)}
    \end{subfigure}
    \caption{$W_{\equil}$, $W_{\textrm{best}}$, $\POA$ varying with $b$. $M=10, \rho =0.2$.}
    \label{fig:trend}
\end{figure}

\begin{restatable}{theorem}{aogthm}
\label{thm:aog}
If $-2+\rho < \mu < -\rho$, then for any level $b$ strategy space $\strategyLevel{b}$, all PNEs have the same social welfare \begin{equation*}
    W_{\equil}(b) =
    \begin{cases}
    (2M + 1)(1 + \mu) + 2(1-\rho) + (2\rho - \mu - 1)b & \textrm{if } b \leq 2M+1 \\
    (4M+4)\rho & \textrm{if } b \geq 2M+2
    \end{cases}.
\end{equation*}
\end{restatable}

The full proof is lengthy and involves analyses of many different cases (\cref{append:aog}). We present the main idea here.

\begin{proofidea}
We make three arguments for this proof:
\begin{enuminline}
    \item Any function in a PNE must satisfy some specific form indicating where it can switch lines
    \item Any PNE under $\strategyLevel{b}$ must consist of only functions that use exactly $b$ segments
    \item For any function with $b$ segments that satisfies the specific form, the optimal strategy for the other player always achieves the same payoff
\end{enuminline}.
\end{proofidea}

\paragraph{Remark}
$-2+\rho < \mu < -\rho$ is in fact a necessary and sufficient condition for all PNEs having the same social welfare for any $b$ and $M$ (see \cref{append:aog-nec}).

Depending on the sign of $2\rho-\mu-1$, $W_{\equil}(b)$ can increase, stay the same, or decrease as $b$ increases until $b=2M+1$. $W_{\equil}(b)$ always decreases at $b=2M+2$ and stays the same afterwards. \Cref{fig:trend} visualizes this trend.
\Cref{fig:demo_nash} summarizes how the characteristics of the PNEs varies in the $\rho$-$\mu$ space.

\paragraph{Price of Anarchy}
The price of anarchy (POA) \citep{koutsoupias1999worst} is the ratio between the best social welfare achieved by any centralized solution and the worst welfare at equilibria: $\POA(b) = \frac{W_{\textrm{best}}(b)}{W_{\equil}(b)}$. We can show that the best centralized social welfare is $W_{\best}(b) = 2M+2 + \mu\cdot \max(2M+1-b, 0)$ (\cref{append:best}). Hence
\begin{equation*}
    \POA(b) =
    \begin{cases}
    1 + \frac{(1-2\rho)(b-1)}{2M+2+2M\mu + (2\rho - \mu - 1)(b-1)} & \text{if } b \leq 2M+1 \\
    \frac{1}{2\rho} & \text{if } b \geq 2M+2
    \end{cases},
\end{equation*}
$\POA(b)$ increases with $b$ up to $b=2M+2$, then stays the same (\cref{fig:trend}).

\paragraph{Interpretation} There are two opposing factors that affect whether increased capability is beneficial for social welfare or not. With increased capability, players can improve their payoff in a non-competitive way (e.g. avoiding mines), which is always beneficial for social welfare; they can also improve payoff in a competitive way (e.g. occupying gold together), which may reduce social welfare. The joint effect of the two factors determines the effect of increasing capability.

\begin{figure}[hbt]
    \centering
    \includegraphics[width=0.65\textwidth]{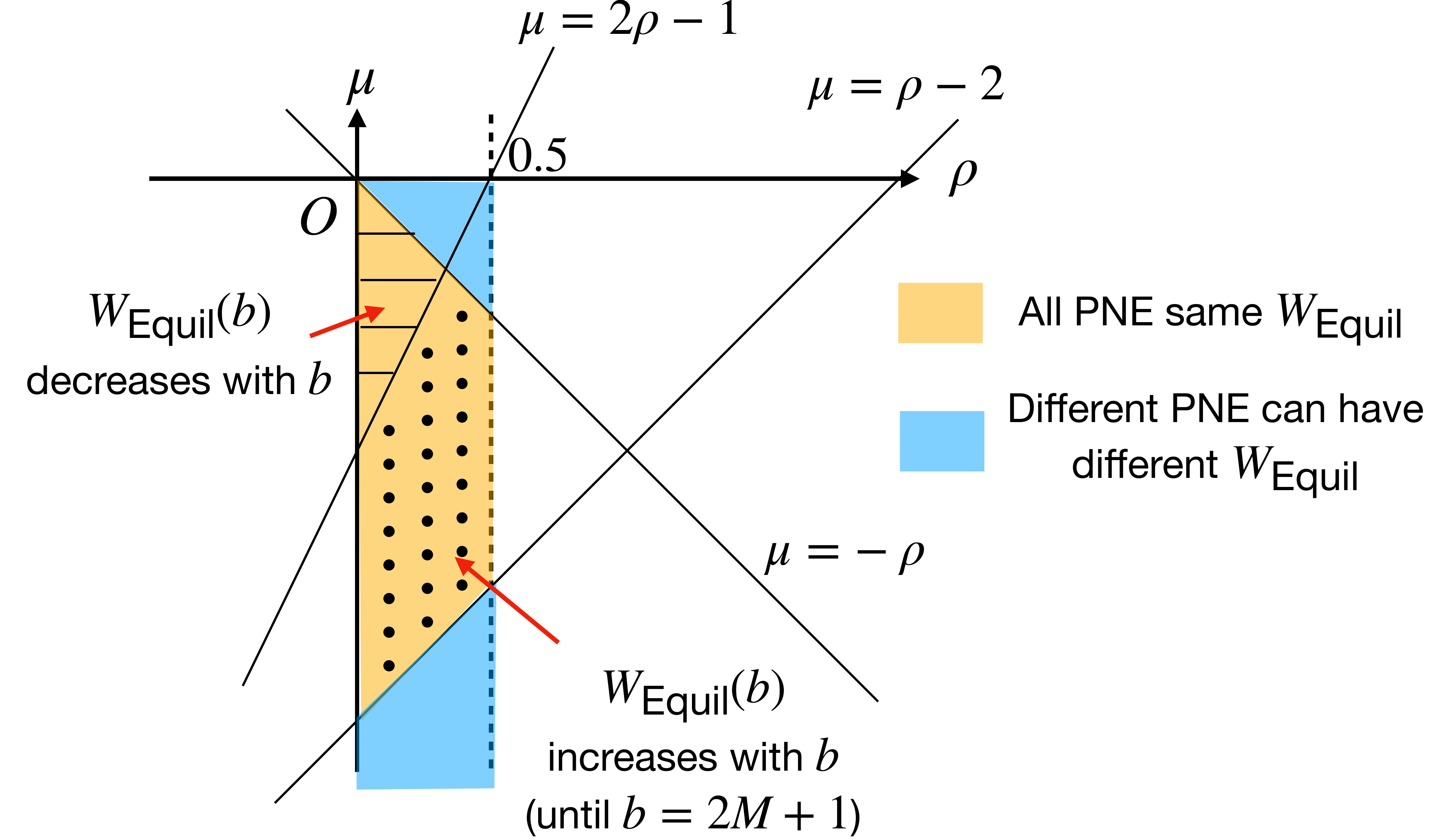}
    \caption{Characteristics of PNEs over the $\rho$-$\mu$ landscape.}
    \label{fig:demo_nash}
\end{figure}

\newpage
\bibliographystyle{splncs04nat}
\bibliography{refs}

\newpage

\appendix
\section{Proofs of the {\mdseries\dbncgda} hardness results}
\label{append:dbncg-da-proof}
% f{{{

\subsection{Proof of \cref{thm:dbncg-da-pls-complete}}

\ThmDbncgDaPlsComplete*
\begin{proof}
    \begin{figure}[hbt]
        \centering
        \begin{subfigure}[t]{.7\textwidth}
            \centering
            \def\figdatatxt{%
                12 11 7 1
                1 2 6 12
            }
            \begin{adjustbox}{width=0.66\textwidth}
                \readarray\figdatatxt\figdata[2,4]

\begin{tikzpicture}[
    node distance={4em},
    main/.style = {draw, circle, minimum size=2em, font=\small},
    edgelabel/.style = {midway, above, sloped, font=\small},
    jumpedge/.style = {gray, dashed}
]

    % uncomment the last few commands to recompute the coordinates
    \clip (-10.5em,-30em) rectangle (19em, 7.5em);

    \tikzset{>=Latex}

    \node[main] (11) {};
    \node[main] (21) [below of=11] {$r_{12}$};
    \node[main] (22) [right of=21] {};
    \node[main] (31) [below of=21] {$r_{13}$};
    \node[main] (32) [right of=31] {$r_{23}$};
    \node[main] (33) [right of=32] {};
    \node[main] (41) [below of=31] {$r_{14}$};
    \node[main] (42) [right of=41] {$r_{24}$};
    \node[main] (43) [right of=42] {$r_{34}$};
    \node[main] (44) [right of=43] {};

    \foreach \i in {1,...,4} {
        \node[main] (s\i) [left of=\i1] {$s_{\i}$};
        \draw[->] (s\i) to (\i1);
    }

    \foreach \i in {1,...,4} {
        \node[main] (t\i) [below of=4\i] {$t_{\i}$};
        \draw[->] (4\i) to (t\i);
    }

    \node[main] (s) at ([xshift=-5em]$(s1)!0.5!(s4)$) {$s$};
    \node[main] (t) at ([yshift=-4em]$(t1)!0.5!(t4)$) {$t$};

    \foreach \y in {2,...,4}
        \foreach \x [evaluate={\xprev=\x-1}] in {2,...,\y} {
            \pgfmathtruncatemacro{\yprev}{\y-1}
            \draw[->] (\y\xprev) to node[edgelabel]{\y} (\y\x);
        }

    \foreach \y [evaluate={\yprev=\y-1}] in {2,...,4}
        \foreach \x in {1,...,\yprev} {
            \pgfmathtruncatemacro{\yprev}{\y-1}
            \draw[->] (\yprev\x) to (\y\x);
        }

    \foreach \out/\i in {90/1,45/2,-45/3,-90/4} {
        \draw[->, out=\out, in=180] (s) to node[edgelabel]{\figdata[1,\i]} (s\i);
        \draw[->] (t\i) to (t);
    }

    \draw[->, jumpedge, out=45, in=45, looseness=1.7]
        (s1) to node[edgelabel, gray, pos=0.55]
        {$r_1;\; w_1=\figdata[2,1]$} (t1);

    \draw[->, jumpedge, out=120, in=45, looseness=3.8]
        (s2) to node[edgelabel, gray, pos=0.45]
        {$r_2;\; w_2=\figdata[2,2]$} (t2);

    \draw[->, jumpedge]
        (s3) to[out=-140, in=-90, looseness=1.4]
        node[edgelabel, gray, pos=0.5]
        {$r_3;\; w_3=\figdata[2,3]$}
        ++($(t4.east) + (9em, 0)$)
        to[out=90, in=30, looseness=2] (t3);

    \draw[->, jumpedge, out=-80, in=-30, looseness=2.3]
        (s4) to node[edgelabel, gray, pos=0.5]
        {$r_4;\; w_4=\figdata[2,4]$} (t4);

    %\draw[gray,step=5em] (-10em,-30em) grid (20em, 20em);
    %\draw [brown] (-10.5em,-30em) rectangle (19em, 7.5em);
    %\clip (-11em,-30em) rectangle (20em, 8.5em);
    %\draw [brown] (current bounding box.south west) rectangle (current bounding box.north east);
\end{tikzpicture}

% vim: tw=80 filetype=tex foldmethod=marker foldmarker=f{{{,f}}} spell spelllang=en
            \end{adjustbox}
            \caption{The graph structure}
        \end{subfigure}
        \begin{subfigure}[t]{.29\textwidth}
            \centering
            \begin{adjustbox}{width=\textwidth}
                \begin{tikzpicture}[
    main/.style = {draw, circle, font=\small, minimum size=2em},
    edgelabel/.style = {midway, above, sloped, font=\small}
]

    \tikzset{>=Latex}

    \node[main] (u) at (-6em,0) {$u$};
    \node[main] (v) at (6em,0) {$v$};
    \node[main] (x) at (0,4em) {$u'$};

    \draw[->] (u) to node[edgelabel]{0, $r_{\infty}$} (x);
    \draw[->] (x) to node[edgelabel]{0} (v);
    \draw[->] (u) to node[edgelabel]{1} (v);

\end{tikzpicture}

% vim: tw=80 filetype=tex foldmethod=marker foldmarker=f{{{,f}}} spell spelllang=en
            \end{adjustbox}
            \caption{A gadget to add a default action for $u$}
        \end{subfigure}
        \caption{
            Illustration of the \dbncgda{} instance corresponding to a
            four-player quadratic threshold game. The distance bound is $b=13$.
            Non-zero-length edges have labels to indicate their lengths.
            \label{fig:dbncg-da-pls-proof}
        }
    \end{figure}
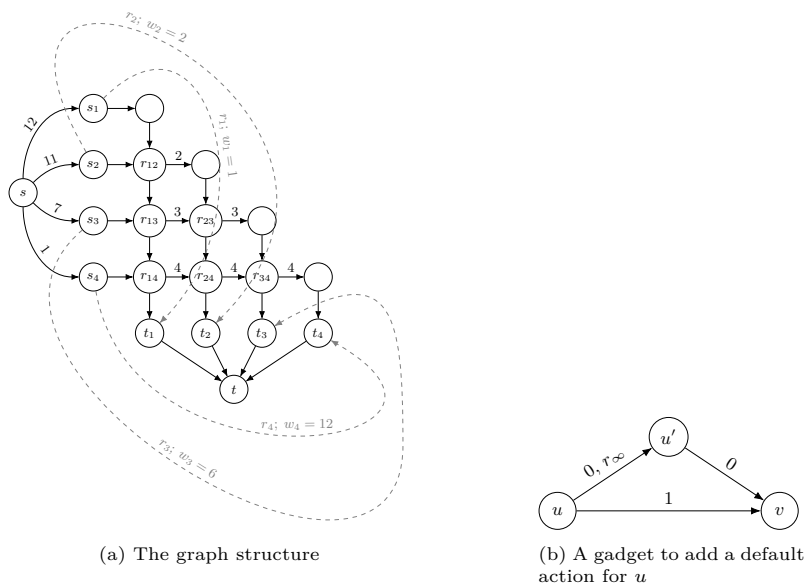

    The best response of a player in \dbncgda{} can be computed in polynomial
    time similarly to \dnc. Now we prove its \PLS-completeness by presenting a
    reduction from the quadratic threshold game. We modify the network layout in
    the proof of \cref{thm:dbncg-pls-complete} as follows. We assign zero length
    to the vertical edges $(v_{ij},\,v_{i+1,j})$, edges for $r_{ij}$, and edges
    in the set $\cup_{1\le i \le n}\{(s_i,\,v_{i1}),\, (v_{ni},\,t_i),\, (t_i,\,
    t)\}$. We also redefine the lengths of some other edges: $w_{(s_i,\, t_i)}
    =i(i-1) + \indicator{i=1} $ and $w_{(s,\,s_i)} = b - w_{(s_i,\, t_i)}$. The
    distance bound is $b=n(n-1) +1$. Note that after replacing weighted edges
    with a chain of unit-length edges to build the \dbncgda{} network, some
    vertices will only have one unit-length outgoing edge, which violates the
    requirements of default action. In this case, we add auxiliary vertices with
    zero-length edges and a resource $r_{\infty}$ that has sufficiently large
    delay to disincentivize any player from taking the zero-length auxiliary
    edges. \Cref{fig:dbncg-da-pls-proof} illustrates this construction. In this
    \dbncgda{} instance, player $i$ will choose edges $(s,\,s_i)$ and $(t_i,\,t)
    $. They then choose between the right-down path from $s_i$ to $t_i$ or the
    edge $(s_i, \,t_i)$ which correspond to the two strategies in the quadratic
    threshold game respectively.
\end{proof}

\subsection{Proof of \cref{thm:dbncg-da-best-np-hard} and
    \cref{thm:dbncg-da-worst-np-hard}}

\ThmDbncgDaBestNpHard*
\begin{proof}
    We adopt the reduction given in the proof of \cref{thm:dbncg-best-np-hard}.
    We modify the edge weights as $w_{(t_i,\,s_{i+1})}=1$, $w_{(s_i,\,t_i)}
    =a_i-1$, and $w_{(t_i,\,t_{i+1})}=0$, as illustrated in
    \cref{fig:dbncg-da-best-hard}. The distance bound is $b=T$. All other
    argument in the original proof follows.

    \begin{figure}[hbt]
        \centering
        \def\edgeVert{-1}
        \def\edgeDiag{1}
        \def\edgeHori{}
        \begin{tikzpicture}[
    node distance={7em},
    nodetext/.style = {minimum size=2em, font=\footnotesize},
    main/.style = {draw, nodetext, circle},
    edgelabel/.style = {midway, above, sloped, font=\footnotesize},
    fastedge/.style = {thin, dashed},
    slowedge/.style = {double, double distance=1pt},
    dot/.style = {minimum width=4em, minimum height=2em}
]

    \tikzset{>=Latex}

    \node[main] (t0) {$t_0$};
    \node[main] (t1) [right of=t0] {$t_1$};
    \node[main] (s1) [yshift=-1em][above of=t1] {$s_1$};
    \node[main] (t2) [right of=t1] {$t_2$};
    \node[main] (s2) [yshift=-1em][above of=t2] {$s_2$};
    \node[dot] (tdot) [right of=t2] {};
    \node[dot] (sdot) [yshift=-1em][above of=tdot] {};
    \node[dot] (dispdot) at ($(tdot)!0.5!(sdot)$) {$\cdots$};
    \node[main] (t3m) [right of=tdot] {$t_{3m}$};
    \node[main] (s3m) [yshift=-1em][above of=t3m] {$s_{3m}$};

    \foreach \i in {1,2,3m}
        \draw[->, fastedge] (s\i) to node[edgelabel]{$a_{\i}\edgeVert$} (t\i);

    \foreach \t / \s / \e in {0/1/, 1/2/, 2/dot/.south west, dot/3m/}
        \draw[->] (t\t.15) to node[edgelabel]{\edgeDiag} (s\s\e);

    \foreach \tp/\tn in {0/1, 1/2, 2/dot, dot/3m}
        \draw[->, slowedge] (t\tp) to node[edgelabel]{\edgeHori} (t\tn);
\end{tikzpicture}

% vim: tw=80 filetype=tex foldmethod=marker foldmarker=f{{{,f}}} spell spelllang=en
        \caption{ \dbncgBestHardCaption{\dbncgda} \label{fig:dbncg-da-best-hard}
        }
    \end{figure}
\end{proof}

With similar adoption of the proof of \cref{thm:dbncg-worst-np-hard}, we can
also prove \cref{thm:dbncg-da-worst-np-hard}, which claims the \NP-hardness of
computing the worst social welfare at equilibrium.

% f}}}

\section{Detailed proof of \cref{thm:db-pp}}
\label{append:db-pp}

\dbpp*

\begin{proof}
    If $\fd$ is a constant function, the total delay achieved by a strategy is
    not affected by the load condition of each edge (thus not affected by other
    players' strategies). So each player's strategy in any PNE is the one in
    $\strategyLevel{b}$ that minimizes the total delay under the game layout.
    Denote this minimum delay as $\delta(b)$. For any $b \ge 1$, we have
    $\strategyLevel{b} \subseteq \strategyLevel{b+1}$, so $\delta(b) \geq
    \delta(b+1)$. And for any $\vs \in \equil(b)$, $W(\vs) = - n \delta(b)$.
    Hence $\bestw{b} \leq \worstw{b+1}$.

    If $\fd$ is not a constant function, we show that there exists an instance
    of \dbncgdas{} with delay function $\fd$ that is not \pp{}. Define $v =
    \min\condSet{x}{d(x) \neq d(x+1)}$. It follows that $d(v') = d(v)$ for all
    $v'\leq v$. We consider the cases $d(v) = 0$ and $d(v) > 0$ separately.
    \\
    \\
    \noindent\textbf{Case 1: $d(v) > 0$}\hspace{1em}
    Denote $\rho = \frac{d(v+1)}{d(v)}$. Since $\fd$ is non-decreasing, $\rho >
    1$. We construct a game with the network layout in \cref{fig:da-pp-pos}
    with $n=v+1$ players. This game is a counterexample for the \pp{} property
    with $b=1$ (i.e., $\bestw{1} > \worstw{2}$).

    First, it is easy to see that the PNEs when $b=1$ are $a$ players take the
    upper path and $v+1-a$ players take the lower path, where $1\leq a\leq v$.
    All PNEs achieves a social welfare of $W_1 = - (v+1) (N_1+N_2+3)d(v)$.

    We set the constants $N_1 = \lfloor \frac{1}{\rho - 1}\rfloor$ and
    $N_2 = \lfloor (N_1+2)\rho\rfloor - 1$, which ensures $N_1 > \frac{1}
    {\rho - 1} - 1$ and $(N_1+2)\rho - 2 < N_2 \leq (N_1+2)\rho - 1$. We claim
    that one PNE when $b=2$ is that all players choose the path from upper left
    to lower right using the switching edge in the middle. Under this strategy
    profile, each player has a total delay of $\delta = (2N_1 + 3) d(v+1)$. This
    is an equilibrium because if any player changes their strategy to the upper or
    the lower horizontal path (the only two alternative strategies) the new
    total delay is $\delta' = (N_1+1)d(v+1) + (N_2+2)d(v) > \delta$ because
    $\frac{\delta' - \delta}{d(v)} = N_2 - \big((N_1+2)
    \rho - 2\big) > 0$. The social welfare is $W_2 = - (v+1) (2N_1 + 3)d(v+1)$.
    Note that
    \begin{align*}
        \frac{W_1 - W_2}{(v+1)d(v)} &= (N_1+1)(\rho-1) - \big(N_2 + 2 - (N_1+2)\rho\big) \\
            &> (\frac{1}{\rho - 1} - 1 + 1)(\rho-1) -
                \big( (N_1+2)\rho - 1 + 2- (N_1+2)\rho\big) = 0.
    \end{align*}
    Hence $\bestw{1} \ge W_1 > W_2 \ge \worstw{2}$.
    \\
    \\
    \noindent\textbf{Case 2: $d(v) = 0$}\hspace{1em}
    We construct a game with the network layout in \cref{fig:da-pp-zero} where
    there are $2v$ players. With $b=1$, half of the players choose
    the upper path and the others choose the lower path, which has a social
    welfare $W_1=0$.

    With the bound $b=2$, we consider this strategy:
    \begin{enuminline}
        \item $v$ players take the path $(s,\, N_1 \text{ edges},\, \text{lower
            right } N_2 \text{ edges},\, t)$
        \item the other $v$ players take the path $(s,\, \text{lower left } N_2
            \text{ edges},\, N_1 \text{ edges},\, t)$
    \end{enuminline}.
    We choose $N_1$ and $N_2$ to be positive integers that satisfy $\frac{N_2+1}
    {N_1} > \frac{d(2v)}{d(v+1)}$. There are $2v$ players occupying the $N_1$
    edges that incur a delay of $N_1d(2v)$ on each player. The social welfare
    $W_2=-2vN_1d(2v) < W_1$. The above strategy is an equilibrium because for
    each player the alternative strategy to avoid the $N_1$ congestion edges is
    to take the lower path $(s,\, N_2 \text{ edges},\, N_2 \text{ edges},\, t)$
    which has a delay of $(N_2+1)d(v+1) > N_1d(2v)$.
    Hence $\bestw{1} \ge W_1 > W_2 \ge \worstw{2}$.
\end{proof}

\section{Representing \gam{} as {\mdseries\dbncgda}}
\label{append:repr}
% f{{{

To represent a \gam{} as a \dbncgda{}, we first order all resources by increasing order in $x_e$, such that $x_1 < x_2 < \dots < x_{|\setE|}$.
Then the network in the corresponding \dbncgda{} has vertices:
\begin{itemize}
    \item $v_{i,j}$ at $(\frac{x_i + x_{i+1}}{2}, j)$ for all $i=1,\dots,|\setE|-1$ and $j=0,\dots,K-1$.
    \item $v_{0,j}$ at $(x_1-1, j)$ and $v_{|\setE|,j}$ at $(x_{|\setE|}+1, j)$ for all $j=0,\dots,K-1$.
    \item A source $s$ and a sink $t$.
\end{itemize}
The positions of the vertices are only to help understanding. The edges are:
\begin{itemize}
    \item $(v_{i,j}, v_{i+1,j})$ for all $i=0,\dots,|\setE-1|$ and $j=0,\dots,K-1$, with length $w_e = 0$. If $y_{i+1} = j$, then the delay function $d_e = - r_{\alpha_{i+1}}$; otherwise, $d_e = 0$.
    \item $(v_{i,j}, v_{i, j'})$ for all $i=1,\dots,|\setE-1|$ and $(j,j')\in \{0,\dots,K-1\}^2$, with $w_e=1$ and $d_e=0$.
    \item $(s, v_{0,j})$ and $(v_{|\setE|,j}, t)$ for all $j=0,\dots,K-1$, with $w_e=1$ and $d_e=0$.
\end{itemize}
The distance bound in the corresponding \dbncgda{} is the same as the maximum number of segments in \gam{}.

\section{Proof of \cref{thm:gm-pp}}
\label{append:gm-pp}

\gmpp*

\begin{proof}
If both $r_g$ and $r_m$ are constant functions, the same proof as in \cref{thm:db-pp} applies to show that \pp{} holds universally here.

If $r_g$ and $r_m$ are not both constant functions, we show that there is always some instance of \gam{} with payoff $r_g$ and $r_m$ that is not \pp{}.

If $r_g$ is not a constant function, let $v = \min \condSet{x}{r_g(x)\neq r_g(x+1)}$. Denote $\rho = \frac{r_g(v+1)}{r_g(v)}$, then $\rho \neq 1$, and $r_g(v') = r_g(v)$ for all $v'\leq v$.
We show that we can always construct a game with $n=v+1$ players and $K=2$ lines that is a counterexample for the \pp{} property with $b=1$ (i.e. $\bestw{1} > \worstw{2}$).
\\
\\
\noindent\textbf{Case 1. $\rho < 1$}\hspace{1em}
The layout of the constructed game is shown in \cref{fig:gm-pp-squeeze}.
All the resources are gold. Block $B$ is constructed as follows.
Following the increasing order of $x$, $y_t=k$ means the $t$-th point is on line $y=k$. $N_0(t)$ ($N_1(t)$) denotes the number of points on line $y=0$ ($y=1$) within the first $t$ points. Denote $D(t) = \rho N_0(t) - N_1(t)$. We use the following algorithm to position the gold:
\begin{enumerate}
    \item \texttt{while} $N_0(t) < \frac{3}{1-\rho}$, do:
    \begin{itemize}
        \item[] If $D(t) \leq 1$, put $y_{t+1} = 0$; else, put $y_{t+1} = 1$;
        \item[] $t\leftarrow t+1$;
    \end{itemize}
    \item \texttt{while} $D(t) \leq 1+\rho$, do: $y_{t+1} = 0, t\leftarrow t+1$.
\end{enumerate}
Denote the total number of gold put down as $N$. Under the above construction, the following properties hold:
\begin{itemize}
    \item For all $t=1\dots N$, $D(t) > 0$ (all prefixes are better)
    \item For all $t=0\dots N-1$, $D(N) - D(t) > 0$ (all suffixes are better)
    \item $D(N) < 3$
    \item $N_0(N) > \frac{3}{1-\rho}$
\end{itemize}
Block $B'$ is obtained by flipping $B$ in both $x$ and $y$ direction.

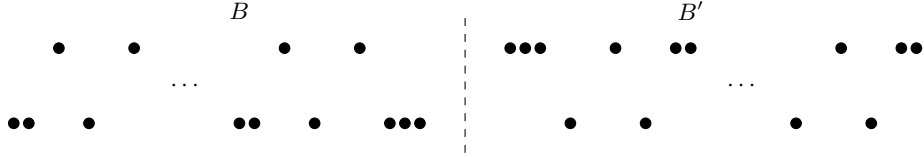
\begin{figure}[hbt]
    \centering
    \begin{tikzpicture}[
    dot/.style = {minimum width=2em, minimum height=2em}
]

    \foreach \basex/\basey/\dir in {0/0/1, 12/1/-1} {
        \foreach \i in {0,3} {
            \filldraw[black] (\basex+\dir*\i, \basey) circle (2pt) node{};
            \filldraw[black] (\basex+\dir*\i+\dir*0.2, \basey) circle (2pt) node{};
            \filldraw[black] (\basex+\dir*\i+\dir*0.6, 1-\basey) circle (2pt) node{};
        }
        
        \foreach \i in {1,4} {
            \filldraw[black] (\basex+\dir*\i, \basey) circle (2pt) node{};
            \filldraw[black] (\basex+\dir*\i+\dir*0.6, 1-\basey) circle (2pt) node{};
        }
        
        \filldraw[black] (\basex+\dir*5, \basey) circle (2pt) node{};
        \filldraw[black] (\basex+\dir*5.2, \basey) circle (2pt) node{};
        \filldraw[black] (\basex+\dir*5.4, \basey) circle (2pt) node{};
        
        \node[dot] (ddd) at (\basex+\dir*2.3,0.5) {$\cdots$};
    }
    
    \filldraw[dashed] (6, 1.4) -- (6, -0.4);
    
    \node[dot] (b) at (3, 1.5) {$B$};
    \node[dot] (b) at (9, 1.5) {$B'$};
    
\end{tikzpicture}
    \caption{
        Counterexample for the case $\rho < 1$ for $r_g$ and $\rho > 1$ for $r_m$.
        Each dot is a resource. The upper (lower) line is $y=1$ ($y=0$). Block $B$ and $B'$ are centrosymmetric. The direction of increasing $x$ is from left to right.
        \label{fig:gm-pp-squeeze}
    }
\end{figure}

For $b=2$, a PNE is all players choose $y=0$ in $B$ and $y=1$ in $B'$. This is a PNE because of the prefix and suffix properties. Its social welfare $W_2 = (v+1) \cdot 2N_0(N)r_g(v+1)$. For $b=1$, the PNEs are $a$ players chooses $y=0$ and $v+1-a$ players choose $y=1$, where $1\leq a \leq v$. All PNEs have the same social welfare $W_1 = (v+1) \cdot (N_0(N)+N_1(N))r_g(v)$. Then
\small
\begin{align*}
    \frac{W_1 - W_2}{(v+1)r_g(v)} &= N_0(N)+N_1(N) - 2\rho N_0(N) =  N_0(N)(1-\rho) - D(N) > 3 - 3 = 0.
\end{align*}
\normalsize
Therefore $W_1 > W_2$, which implies $\bestw{1} > \worstw{2}$.
\\
\\
\noindent\textbf{Case 2. $\rho > 1$}\hspace{1em}
The layout of the constructed game is in \cref{fig:gm-pp-stuck}. We choose $N > \frac{\rho}{\rho-1}$.
A PNE for $b=2$ is all players choose the first $N$ gold on $y=0$ and the $N$ gold on $y=1$, which has welfare $W_2 = (v+1) \cdot 2N r_g(v+1)$. A PNE for $b=1$ is all players choose $y=0$, which has welfare $W_1 = (v+1) \cdot (2N+1) r_g(v+1)$. Clearly $W_1 > W_2$, so $\bestw{1} > \worstw{2}$.
\\

\begin{figure}[hbt]
    \centering
    \begin{tikzpicture}[
    dot/.style = {minimum width=2em, minimum height=2em}
]

    \foreach \x/\y in {0/0, 2.5/1, 5/0} {
        \foreach \dx in {0, 0.5, 2} {
            \filldraw[black] (\x + \dx, \y) circle (2pt) node{};
        }
        \node[dot] (dd) at (\x+1.3, \y) {$\cdots$};
    }
    
    \draw [decorate,decoration={brace,mirror,amplitude=3mm,raise=2mm}] (0, 0) -- (2, 0);
    \draw [decorate,decoration={brace,mirror,amplitude=3mm,raise=2mm}] (5, 0) -- (7, 0);
    \draw [decorate,decoration={brace,amplitude=3mm,raise=2mm}] (2.5, 1) -- (4.5, 1);
    \node[dot] (b) at (1, -0.7) {$N$};
    \node[dot] (b) at (6, -0.7) {$N+1$};
    \node[dot] (b) at (3.5, 1.7) {$N$};
    
\end{tikzpicture}
    \caption{
       Counterexample for the case $\rho > 1$ for $r_g$ and $\rho < 1$ for $r_m$.
        \label{fig:gm-pp-stuck}
    }
\end{figure}
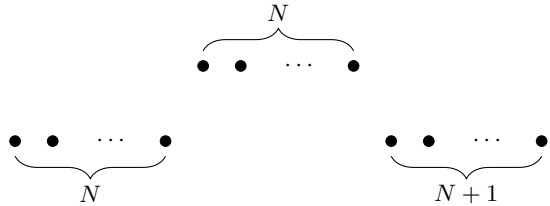

If $r_m$ is not a constant function, let $v = \min \condSet{i}{r_m(i)\neq r_m(i+1)}$. Denote $\rho = \frac{r_m(v+1)}{r_m(v)}$.
\\
\\
\noindent\textbf{Case 1. $\rho < 1$}\hspace{1em}
We use the same layout as in \cref{fig:gm-pp-stuck} with all resources being mines and $v+1$ players. We choose $N > \frac{\rho}{1-\rho}$. A PNE for $b=2$ is all players chooses the bottom-right $N+1$ mines and avoiding all other mines, which has welfare $W_2 = (v+1) \cdot (N+1) r_m(v+1)$.
A PNE for $b=1$ is all players choose $y=1$, which has welfare $W_1 = (v+1) \cdot N r_m(v+1)$. Clearly $W_1 > W_2$.
\\
\\
\noindent\textbf{Case 2. $\rho > 1$}\hspace{1em}
We use the same layout as in \cref{fig:gm-pp-squeeze} with all resources being mines and $v+1$ players. Block $B$ is constructed in a similar way with $D(t) = \frac{1}{\rho}N_0(t) - N_1(t)$, and the algorithm is:
\begin{enumerate}
    \item \texttt{while} $N_1(t) < \frac{3}{1-1/\rho}$, do:
    \begin{itemize}
        \item[] If $D(t) \leq 1$, put $y_{t+1} = 0$; else, put $y_{t+1} = 1$;
        \item[] $t\leftarrow t+1$;
    \end{itemize}
    \item \texttt{while} $D(t) \leq 1+ \frac{1}{\rho}$, do: $y_{t+1} = 0, t\leftarrow t+1$.
\end{enumerate}
The properties following the construction becomes:
\begin{itemize}
    \item For all $t=1\dots N$, $D(t) > 0$ (all prefixes are better)
    \item For all $t=0\dots N-1$, $D(N) - D(t) > 0$ (all suffixes are better)
    \item $D(N) < 3$
    \item $N_1(N) > \frac{3\rho}{\rho-1}$
\end{itemize}

For $b=2$, a PNE is all players choose $y=1$ in $B$ and $y=0$ in $B'$, whose social welfare $W_2 = (v+1)\cdot 2N_1(N)r_m(v+1)$. For $b=1$, the PNEs are $a$ players choose $y=0$ and $v+1-a$ players choose $y=1$, where $1\leq a\leq v$. All PNEs have the same social welfare $W_1 = (v+1)\cdot (N_0(N)+N_1(N))r_m(v)$. Then
\small
\begin{align*}
    \frac{W_1 - W_2}{-(v+1)r_m(v)} = 2\rho N_1(N) - N_0(N)-N_1(N)= (\rho - 1)N_1(N) - \rho D(N) > 3\rho - 3\rho = 0.
\end{align*}
\normalsize
Therefore $W_1 > W_2$, which implies $\bestw{1} > \worstw{2}$.

\end{proof}

\section{Proofs for the \aog{}}
\label{append:aog}

\subsection{Preliminaries}
We first establish some notational convenience for the subsequent analyses:
\begin{itemize}
    \item Following the increasing order of $x$, we use $(x_t, y_t, \alpha_t)$ to denote the location and type of the $t$-th resource.

    \item Each player's strategy can be represented as a function over the integer domain $[0, 4M+1]$, specifying the function value at each $x_t$ for $t\in[0, 4M+1]$. We represent a function compactly as the set of intervals $\{[a_{\xi}^-, a_{\xi}^+]\}_{{\xi}=0}^{c-1}$ over $t$ with $f(t)=1$. For example, $f=\{[0, 2], [5, 5] \}$ represents the function of $f(t)=1$ if $t\in [0, 2]$ or $t=5$, and $f(t)=0$ otherwise.
    \item We use the \textbf{canonical} representation of $f$ throughout the paper, where
    \begin{itemize}
        \item $a_{\xi}^-, a_{\xi}^+$ are integers for all $\xi$, and $a_0^-\geq 0, a_{c-1}^+\leq 4M+1$;
        \item $a_{\xi}^-\leq a_{\xi}^+, a_{{\xi}+1}^- - a_{\xi}^+ \geq 2$ for all $\xi$, i.e. the representation uses the least number of intervals.
    \end{itemize}
\item Denote the set of functions with exactly $k$ segments as $\mathcal{F}_k$. Then $\strategyLevel{b} = \bigcup_{k\leq b} \mathcal{F}_k$. The following general form covers all functions within $\mathcal{F}_k$:
    \begin{itemize}
        \item For $k=2c+1$ (odd $k$), $\mathcal{F}_k = \Big\{ \{[a_{\xi}^-, a_{\xi}^+]\}_{{\xi}=0}^c \big| a_0^- = 0, a_c^+=4M+1 \Big\} \bigcup \Big\{ \{[a_{\xi}^-, a_{\xi}^+]\}_{{\xi}=0}^{c-1} \big| a_0^- > 0, a_{c-1}^+ < 4M+1 \Big\}$
        \item For $k=2c$ (even $k$), $\mathcal{F}_k = \Big\{ \{[a_{\xi}^-, a_{\xi}^+]\}_{{\xi}=0}^{c-1} \big| a_0^- = 0, a_{c-1}^+<4M+1 \Big\} \bigcup \Big\{ \{[a_{\xi}^-, a_{\xi}^+]\}_{{\xi}=0}^{c-1} \big| a_0^- > 0, a_{c-1}^+ = 4M+1 \Big\}$
    \end{itemize}

\end{itemize}

\subsection{Proof of \cref{thm:aog}}

\begin{restatable}{lemma}{optima}
\label{lemma:optima}
Given an arbitrary set of strategies used by the other players $\textbf{f}_{-i}$ and any $k_0 \geq 1$, denote $f_i^* = \argmax_{f_i\in \strategyLevel{k_0}} u_i(f_i, \textbf{f}_{-i})$ as the optimal strategy for player $i$, then $f_i^* = \{[a_{\xi}^-, a_{\xi}^+]\}_{{\xi}=0}^{c-1}$ must satisfy the following condition:
\[\forall \xi, a_{\xi}^- = 0 \textrm{ or } 4j_{\xi}^-+3 \textrm{ for some } j_{\xi}^-, a_{\xi}^+ = 4M+1 \textrm{ or } 4j_{\xi}^+ \textrm{ for some } j_{\xi}^+.\]
\end{restatable}

\begin{proof}
This lemma essentially states that the segments of an optimal strategy can only start and end at particular locations in the sequence. We prove this lemma by showing that if $f^*_i \in \mathcal{F}_k$ does not satisfy the above condition, then there exists $f'_i$ which uses $k'\leq k$ segments and achieves a payoff $u'_i = u_i(f'_i, \textbf{f}_{-i})$ higher than $u_i^* = u_i(f^*_i, \textbf{f}_{-i})$, therefore contradicting the fact that $f^*_i$ is optimal.

If $f^*_i$ does not satisfy the given condition, then either there exists a $\xi_0$ where $a_{\xi_0}^- \neq 0$ and $a_{\xi_0}^- \neq4j+3$ for all $j$, or there exists a $\xi_0$ where $a_{\xi_0}^+ \neq 4M+1$ and $a_{\xi_0}^+ \neq4j$ for all $j$. We consider each case separately here.
\begin{enumerate}
    \item There exists a $\xi_0$ where $a_{\xi_0}^- \neq 0$ and $a_{\xi_0}^- \neq4j+3$ for all $j$. Consider the value of $a_{\xi_0}^-$:
    \begin{itemize}
        \item $a_{\xi_0}^- = 4j$ for some $j\neq 0$. Since $f^*_i$ is in canonical form, $f^*_i(4j-1) = 0$. We know $y_{4j-1}=0, \alpha_{4j-1}=\textrm{mine}$. Let $f'_i$ be identical to $f^*_i$ except changing $a_{\xi_0}^- = 4j-1$, then the payoff achieved by $f'_i$ is $u'_i = u^*_i - \mu > u^*_i$. And the number of segments of $f'_i$ $k' \leq k$.
        \item $a_{\xi_0}^- = 4j+1$ for some $j$. We know $y_{4j+1}=0, \alpha_{4j+1}=\textrm{gold}$. Let $f'_i$ be identical to $f^*_i$ except changing $a_{\xi_0}^- = 4j+2$ (if this makes interval $\xi_0$ empty, then remove interval $\xi_0$), then the payoff achieved by $f'_i$ is $u'_i = u^*_i + r_g(x'_{4j+1}) > u^*_i$, and $k' \leq k$.
        \item $a_{\xi_0}^- = 4j+2$ for some $j$. We know $y_{4j+2}=1, \alpha_{4j+2}=\textrm{mine}$. Let $f'_i$ be identical to $f^*_i$ except changing $a_{\xi_0}^- = 4j+3$ (if this makes interval $\xi_0$ empty, then remove interval $\xi_0$), then the payoff achieved by $f'_i$ is $u'_i = u^*_i - \mu > u^*_i$, and $k' \leq k$.
    \end{itemize}

    \item There exists a $\xi_0$ where $a_{\xi_0}^+ \neq 4M+1$ and $a_{\xi_0}^+ \neq4j$ for all $j$. Consider the value of $a_{\xi_0}^+$:
    \begin{itemize}
        \item $a_{\xi_0}^+ = 4j+1$ for some $j< M$. We know $y_{4j+1}=0, \alpha_{4j+1}=\textrm{gold}$. Let $f'_i$ be identical to $f^*_i$ except changing $a_{\xi_0}^+ = 4j$, then the payoff achieved by $f'_i$ is $u'_i = u^*_i + r_g(x'_{4j+1}) > u^*_i$, and $k' \leq k$.
        \item $a_{\xi_0}^+ = 4j+2$ for some $j$. We know $y_{4j+2}=1, \alpha_{4j+2}=\textrm{mine}$. Let $f'_i$ be identical to $f^*_i$ except changing $a_{\xi_0}^+ = 4j+1$, then the payoff achieved by $f'_i$ is $u'_i = u^*_i - \mu > u^*_i$, and $k' \leq k$.
        \item $a_{\xi_0}^+ = 4j+3$ for some $j$. We know $y_{4j+4}=1, \alpha_{4j+4}=\textrm{gold}$. Let $f'_i$ be identical to $f^*_i$ except changing $a_{\xi_0}^+ = 4j+4$, then the payoff achieved by $f'_i$ is $u'_i = u^*_i + r_g(x'_{4j+4}) > u^*_i$, and $k' \leq k$.
    \end{itemize}
\end{enumerate}
All the cases imply that $f^*_i$ cannot be optimal within $\strategyLevel{k_0}$, which is a contradiction. Therefore $f^*_i$ must satisfy the given condition.
\end{proof}

\begin{restatable}{lemma}{form}
\label{lemma:form}
Following \cref{lemma:optima}, if the optimal strategy $f^*_i \in \mathcal{F}_k$ where $k \leq 2M+1$, then it must has the following properties:
\begin{itemize}
    \item If $k=2c+1$ (odd $k$), then $f_i^*$ satisfies either condition $S_1$ or condition $S_2$, and $f_i^*$ always covers $M + c + 1$ gold and $M - c$ mines.
    \begin{itemize}
        \item $S_1$: $f^*_i = \{[a_{\xi}^-, a_{\xi}^+]\}_{{\xi}=0}^c$
        where $a_0^- = 0, a_c^+=4M+1$,
        and $\{a_{\xi}^+ = 4j_{\xi}^+\}_{\xi=0}^{c-1}$,
        $\{a_{\xi}^- = 4j_{\xi}^-+3\}_{\xi=1}^{c}$,
        $\{j_{\xi}^- < j_{\xi}^+\}_{\xi=1}^{c-1}$,
        $\{j_{\xi}^+ \leq j_{\xi+1}^-\}_{\xi=0}^{c-1}$,
        $j_{c}^- < M$, $j_{0}^+\geq 0$
        for some $\{j_{\xi}^-\}_{\xi=1}^{c}$
        and $\{j_{\xi}^+\}_{\xi=0}^{c-1}$.
        \item $S_2$: $f^*_i = \{[a_{\xi}^-, a_{\xi}^+]\}_{{\xi}=0}^{c-1}$
        where $\{a_{\xi}^- = 4j_{\xi}^-+3\}_{\xi=0}^{c-1}$,
        $\{a_{\xi}^+ = 4j_{\xi}^+\}_{\xi=0}^{c-1}$,
        $\{j_{\xi}^- < j_{\xi}^+\}_{\xi=0}^{c-1}$,
        $\{j_{\xi}^+ \leq j_{\xi+1}^-\}_{\xi=0}^{c-2}$,
        $j_{c-1}^+ \leq M$, $j_{0}^-\geq 0$
        for some $\{j_{\xi}^-\}_{\xi=0}^{c-1}$
        and $\{j_{\xi}^+\}_{\xi=0}^{c-1}$.
    \end{itemize}

    \item If $k=2c$ (even $k$), then either $f_i^*$ satisfies condition $S_3$, in which case $f_i^*$ always covers $M+c+1$ gold and $M-c+1$ mines, or $f_i^*$ satisfies condition $S_4$, in which case $f_i^*$ always covers $M+c$ gold and $M-c$ mines.
    \begin{itemize}
        \item $S_3$: $f^*_i = \{[a_{\xi}^-, a_{\xi}^+]\}_{{\xi}=0}^{c-1}$
        where $a_0^- = 0$,
        and $\{a_{\xi}^+ = 4j_{\xi}^+\}_{\xi=0}^{c-1}$,
        $\{a_{\xi}^- = 4j_{\xi}^-+3\}_{\xi=1}^{c-1}$,
        $\{j_{\xi}^- < j_{\xi}^+\}_{\xi=1}^{c-1}$,
        $\{j_{\xi}^+ \leq j_{\xi+1}^-\}_{\xi=0}^{c-2}$,
        $j_{c-1}^+ \leq M$, $j_{0}^+\geq 0$
        for some $\{j_{\xi}^-\}_{\xi=1}^{c-1}$
        and $\{j_{\xi}^+\}_{\xi=0}^{c-1}$.
        \item $S_4$: $f^*_i = \{[a_{\xi}^-, a_{\xi}^+]\}_{{\xi}=0}^{c-1}$
        where $a_{c-1}^+ = 4M+1$,
        and $\{a_{\xi}^+ = 4j_{\xi}^+\}_{\xi=0}^{c-2}$,
        $\{a_{\xi}^- = 4j_{\xi}^-+3\}_{\xi=0}^{c-1}$,
        $\{j_{\xi}^- < j_{\xi}^+\}_{\xi=0}^{c-2}$,
        $\{j_{\xi}^+ \leq j_{\xi+1}^-\}_{\xi=0}^{c-2}$,
        $j_{0}^- \geq 0$, $j_{c-1}^- < M$
        for some $\{j_{\xi}^-\}_{\xi=0}^{c-1}$
        and $\{j_{\xi}^+\}_{\xi=0}^{c-2}$.
    \end{itemize}
\end{itemize}

\end{restatable}

\begin{proof}
We prove for each case.

\noindent\textbf{1. $k=2c+1$ (odd $k$)}

Here, $f^*_i$ is either of form $\{[a_{\xi}^-, a_{\xi}^+]\}_{{\xi}=0}^c$ where $a_0^- = 0$ and $a_c^+=4M+1$, or of form $\{[a_{\xi}^-, a_{\xi}^+]\}_{{\xi}=0}^{c-1}$ where $a_0^- > 0$ and $a_{c-1}^+ < 4M+1$.

\paragraph{1.1}
If $f^*_i$ is of form $\{[a_{\xi}^-, a_{\xi}^+]\}_{{\xi}=0}^c$ where $a_0^- = 0$ and $a_c^+=4M+1$, then according to \cref{lemma:optima}, we have $\{a_{\xi}^+ = 4j_{\xi}^+\}_{\xi=0}^{c-1}$ and $\{a_{\xi}^- = 4j_{\xi}^-+3\}_{\xi=1}^{c}$ for some $\{j_{\xi}^-\}_{\xi=1}^{c}$ and $\{j_{\xi}^+\}_{\xi=0}^{c-1}$. And since $f^*_i$ is canonical, we have $\{j_{\xi}^- < j_{\xi}^+\}_{\xi=1}^{c-1}$, $\{j_{\xi}^+ \leq j_{\xi+1}^-\}_{\xi=0}^{c-1}$, $j_{c}^- < M$, $j_{0}^+\geq 0$. Therefore, in this case, $f^*_i$ satisfies condition $S_1$.

Now consider the number of gold and mines covered by $f^*_i$ that satisfies condition $S_1$. First consider $c\geq 1$.
Segment $[a_{\xi}^-, a_{\xi}^+]$ (where $f^*_i$ has value 1) covers $j_{\xi}^+ - j_{\xi}^-$ gold and $j_{\xi}^+ - j_{\xi}^- - 1$ mines, for $\xi=1,\dots,c-1$.
Segment $[a_{\xi-1}^+ + 1, a_{\xi}^- - 1]$ (where $f^*_i$ has value 0) covers $j_{\xi}^- - j_{\xi-1}^+ + 1$ gold and $j_{\xi}^- - j_{\xi-1}^+$ mines, for $\xi=1,\dots,c$.
Segment $[a_{0}^-, a_{0}^+]$ covers $j_0^+ + 1$ gold and $j_0^+$ mines, and segment $[a_{c}^-, a_{c}^+]$ covers $M - j_c^-$ gold and $M - j_c^- - 1$ mines.
Therefore, the number of gold covered by $f_i^*$ is
\begin{equation*}
    j_0^+ + 1 + \sum_{\xi=1}^{c-1} (j_{\xi}^+ - j_{\xi}^-) + \sum_{\xi=1}^{c} (j_{\xi}^- - j_{\xi-1}^+ + 1) + M - j_c^- = M + c + 1,
\end{equation*}
and the number of mines covered by $f_i^*$ is
\begin{equation*}
    j_0^+ + \sum_{\xi=1}^{c-1} (j_{\xi}^+ - j_{\xi}^- - 1) + \sum_{\xi=1}^{c} (j_{\xi}^- - j_{\xi-1}^+) + M - j_c^- - 1 = M - c.
\end{equation*}
It's straightforward to check that the above expressions also hold for $c=0$.

\paragraph{1.2}
If $f^*_i$ is of form $\{[a_{\xi}^-, a_{\xi}^+]\}_{{\xi}=0}^{c-1}$ where $a_0^- > 0$ and $a_{c-1}^+ < 4M+1$, then similar to case 1.1, \cref{lemma:optima} and the fact that $f^*_i$ is canonical imply that $f^*_i$ satisfies condition $S_2$.

Consider the number of gold and mines covered by $f^*_i$ that satisfies condition $S_2$.
Segment $[a_{\xi}^-, a_{\xi}^+]$ (where $f^*_i$ has value 1) covers $j_{\xi}^+ - j_{\xi}^-$ gold and $j_{\xi}^+ - j_{\xi}^- - 1$ mines, for $\xi=0,\dots,c-1$.
Segment $[a_{\xi-1}^+ + 1, a_{\xi}^- - 1]$ (where $f^*_i$ has value 0) covers $j_{\xi}^- - j_{\xi-1}^+ + 1$ gold and $j_{\xi}^- - j_{\xi-1}^+$ mines, for $\xi=1,\dots,c-1$.
Segment $[0, a_0^- - 1]$ (where $f_i^*$ has value 0) covers $j_0^- + 1$ gold and $j_0^-$ mines, and segment $[a_{c-1}^+ + 1, 4M+1]$ (where $f_i^*$ has value 0) covers $M - j_{c-1}^+ + 1$ gold and $M - j_{c-1}^+$ mines.
Therefore, the number of gold covered by $f_i^*$ is
\begin{equation*}
    j_0^- + 1 + \sum_{\xi=0}^{c-1} (j_{\xi}^+ - j_{\xi}^-) + \sum_{\xi=1}^{c-1} (j_{\xi}^- - j_{\xi-1}^+ + 1) + M - j_{c-1}^+ + 1 = M + c + 1,
\end{equation*}
and the number of mines covered by $f_i^*$ is
\begin{equation*}
    j_0^- + \sum_{\xi=0}^{c-1} (j_{\xi}^+ - j_{\xi}^- - 1) + \sum_{\xi=1}^{c-1} (j_{\xi}^- - j_{\xi-1}^+) + M - j_{c-1}^+ = M - c.
\end{equation*}
\\
\textbf{2. $k=2c$ (even $k$)}

Here, $f^*_i$ is either of form $\{[a_{\xi}^-, a_{\xi}^+]\}_{{\xi}=0}^{c-1}$ where $a_0^- = 0$ and $a_{c-1}^+<4M+1$, or of form $\{[a_{\xi}^-, a_{\xi}^+]\}_{{\xi}=0}^{c-1}$ where $a_0^- > 0$ and $a_{c-1}^+ = 4M+1$.

\paragraph{2.1}
If $f^*_i$ is of form $\{[a_{\xi}^-, a_{\xi}^+]\}_{{\xi}=0}^{c-1}$ where $a_0^- = 0$ and $a_{c-1}^+<4M+1$, then similar to case 1.1, \cref{lemma:optima} and the fact that $f^*_i$ is canonical imply that $f^*_i$ satisfies condition $S_3$.

Consider the number of gold and mines covered by $f^*_i$ that satisfies condition $S_3$.
Segment $[a_{\xi}^-, a_{\xi}^+]$ (where $f^*_i$ has value 1) covers $j_{\xi}^+ - j_{\xi}^-$ gold and $j_{\xi}^+ - j_{\xi}^- - 1$ mines, for $\xi=1,\dots,c-1$.
Segment $[a_{\xi-1}^+ + 1, a_{\xi}^- - 1]$ (where $f^*_i$ has value 0) covers $j_{\xi}^- - j_{\xi-1}^+ + 1$ gold and $j_{\xi}^- - j_{\xi-1}^+$ mines, for $\xi=1,\dots,c-1$.
Segment $[a_0^-, a_0^+]$ (where $f_i^*$ has value 1) covers $j_0^+ + 1$ gold and $j_0^+$ mines, and segment $[a_{c-1}^+ + 1, 4M+1]$ (where $f_i^*$ has value 0) covers $M - j_{c-1}^+ + 1$ gold and $M - j_{c-1}^+$ mines.
Therefore, the number of gold covered by $f_i^*$ is
\begin{equation*}
    j_0^+ + 1 + \sum_{\xi=1}^{c-1} (j_{\xi}^+ - j_{\xi}^-) + \sum_{\xi=1}^{c-1} (j_{\xi}^- - j_{\xi-1}^+ + 1) + M - j_{c-1}^+ + 1 = M + c + 1,
\end{equation*}
and the number of mines covered by $f_i^*$ is
\begin{equation*}
    j_0^+ + \sum_{\xi=1}^{c-1} (j_{\xi}^+ - j_{\xi}^- - 1) + \sum_{\xi=1}^{c-1} (j_{\xi}^- - j_{\xi-1}^+) + M - j_{c-1}^+ = M - c + 1.
\end{equation*}

\paragraph{2.2}
If $f^*_i$ is of form $\{[a_{\xi}^-, a_{\xi}^+]\}_{{\xi}=0}^{c-1}$ where $a_0^- > 0$ and $a_{c-1}^+ = 4M+1$, then similar to case 1.1, \cref{lemma:optima} and the fact that $f^*_i$ is canonical imply that $f^*_i$ satisfies condition $S_4$.

Consider the number of gold and mines covered by $f^*_i$ that satisfies condition $S_4$.
Segment $[a_{\xi}^-, a_{\xi}^+]$ (where $f^*_i$ has value 1) covers $j_{\xi}^+ - j_{\xi}^-$ gold and $j_{\xi}^+ - j_{\xi}^- - 1$ mines, for $\xi=0,\dots,c-2$.
Segment $[a_{\xi-1}^+ + 1, a_{\xi}^- - 1]$ (where $f^*_i$ has value 0) covers $j_{\xi}^- - j_{\xi-1}^+ + 1$ gold and $j_{\xi}^- - j_{\xi-1}^+$ mines, for $\xi=1,\dots,c-1$.
Segment $[0, a_0^- - 1]$ (where $f_i^*$ has value 0) covers $j_0^- + 1$ gold and $j_0^-$ mines, and segment $[a_{c-1}^-, a_{c-1}^+]$ (where $f_i^*$ has value 1) covers $M - j_{c-1}^-$ gold and $M - j_{c-1}^- - 1$ mines.
Therefore, the number of gold covered by $f_i^*$ is
\begin{equation*}
    j_0^- + 1 + \sum_{\xi=0}^{c-2} (j_{\xi}^+ - j_{\xi}^-) + \sum_{\xi=1}^{c-1} (j_{\xi}^- - j_{\xi-1}^+ + 1) + M - j_{c-1}^- = M + c,
\end{equation*}
and the number of mines covered by $f_i^*$ is
\begin{equation*}
    j_0^- + \sum_{\xi=0}^{c-2} (j_{\xi}^+ - j_{\xi}^- - 1) + \sum_{\xi=1}^{c-1} (j_{\xi}^- - j_{\xi-1}^+) + M - j_{c-1}^- - 1 = M - c.
\end{equation*}

\end{proof}

\begin{restatable}{lemma}{payoff}
\label{lemma:payoff}
If one player (call it player A) covers $g_A$ gold, player B covers $g_B$ gold and $m_B$ mines, and $g_A + g_B \geq 2M + 2$, then there is an upper bound on player B's payoff:
\begin{equation*}
u_B \leq (1-\rho) (2M+2 - g_A) + \rho g_B + \mu m_B
\end{equation*}
\end{restatable}

\begin{proof}
Among the gold covered by player B, denote the number of them also covered by player A as $d$. Since the total number of gold is $2M+2$, we have $g_A+g_B - d \leq 2M+2$, i.e. $d\geq g_A + g_B - 2M - 2$. Therefore,
\begin{align*}
s_B &= g_B - d + d\rho + m_B\mu \\
&\leq g_B - (1-\rho) (g_A + g_B - 2M - 2) + m_B\mu \\
&= (1-\rho) (2M+2 - g_A) + \rho g_B + \mu m_B.
\end{align*}
\end{proof}

We define that
a strategy achieves \textit{best coverage} if it covers all the gold that is not covered by the other player.

\begin{restatable}{lemma}{cover}
\label{lemma:cover}
Given one player (call it player A) covers $g_A$ gold, if a strategy $f_B$ for player B covers $g_B$ gold and $m_B$ mines and achieves best coverage, then any strategy $f_B'$ that covers $g_B'$ gold and $m_B'$ mines will achieve a lower payoff than $f_B$, if
\begin{equation*}
    g_B' \leq g_B \land m_B' > m_B, \textrm{ or } g_B' < g_B \land m_B' \geq m_B
\end{equation*}
\end{restatable}

\begin{proof}
Since $f_B$ achieves best coverage, it covers $2M+2 - g_A$ gold that is not covered by player A, and $g_B + g_A - 2M -2$ gold that is covered by A. So the payoff achieved by $f_B$ is
\begin{align*}
    u_B &= 2M+2 - g_A + \rho (g_B + g_A - 2M -2) + m_B\mu \\
    &= (1-\rho) (2M+2 - g_A) + \rho g_B + \mu m_B.
\end{align*}
Consider the payoff of $f_B'$. By \cref{lemma:payoff},
\begin{equation*}
    u_B' \leq (1-\rho) (2M+2 - g_A) + \rho g_B' + \mu m_B'.
\end{equation*}
Since $\rho > 0$ and $\mu < 0$, we can see that if $ g_B' \leq g_B \land m_B' > m_B$, or $g_B' < g_B \land m_B' \geq m_B$,
\begin{equation*}
    u_B' < (1-\rho) (2M+2 - g_A) + \rho g_B + \mu m_B = u_B.
\end{equation*}

\end{proof}

\begin{restatable}{lemma}{exact}
\label{lemma:exact}
If both players' strategy space is $\strategyLevel{b}$ ($b\leq 2M + 2$), then for all PNE $(f^*_1, f^*_2)$, $f^*_1, f^*_2\in \mathcal{F}_b$, i.e. both strategies in the equilibria must use exactly $b$ segments.
\end{restatable}

\begin{proof}
We prove by induction.

\paragraph{Base case}
For $b=1$, there is only two possible strategies in $\strategyLevel{1}$: $f^0 = \{\}$ and $f^1 = \{[0, 4M+1]\}$. Both uses exactly 1 segment. So the statement holds.

\paragraph{Induction step}
Consider the case $b=k$. First we show that if one of the strategies in a PNE uses exactly $k$ segments, then the other strategy must also use exactly $k$ segments. Without loss of generality, let $f_1^* \in \mathcal{F}_k$.
\\
\\
\noindent\textbf{1. $k=2c+1$ (odd $k$)}

By \cref{lemma:form}, $f_1^*$ must satisfy condition $S_1$ or $S_2$.

\paragraph{1.1} If $f_1^*$ satisfies condition $S_1$, let $f^*_1 = \{[a_{\xi}^-, a_{\xi}^+]\}_{{\xi}=0}^c$
where $a_0^- = 0, a_c^+=4M+1$,
$\{a_{\xi}^+ = 4j_{\xi}^+\}_{\xi=0}^{c-1}$,
$\{a_{\xi}^- = 4j_{\xi}^-+3\}_{\xi=1}^{c}$, and $f_1^*$ covers $M+c+1$ gold and $M-c$ mines.
We construct $\hat{f}_2 = \{[\hat{a}_{\xi}^-, \hat{a}_{\xi}^+]\}_{{\xi}=0}^{c-1}$ according to $f_1^*$ by setting $\{\hat{a}_{\xi}^- = 4j_{\xi}^+ + 3, \hat{a}_{\xi}^+ = 4 \cdot  \max(j_{\xi+1}^-, j_{\xi}^+ + 1)\}_{\xi=0}^{c-1}$ (note that here $j_{\xi}^+$ and $j_{\xi}^-$ are the values used by $f_1^*$). It's easy to check that $\hat{f}_2$ satisfies condition $S_2$ and covers $g_2=M+c+1$ gold and $m_2=M-c$ mines.
In particular, among the gold covered by $\hat{f}_2$, $2c$ of them are also covered by $f_1^*$, and $M-c+1$ of them are covered by $\hat{f}_2$ only.
Therefore, $\hat{f}_2$ achieves best coverage. $\hat{f}_2$ achieves a payoff of
\begin{equation*}
    \hat{u}_2 = M - c + 1 + 2c\rho + (M - c)\mu.
\end{equation*}

We show here that any $f'_2 \in \mathcal{F}_{k'}$ where $k' < k$ will achieve a payoff $u'_2 < \hat{u}_2$, therefore $f_2^*$ must use exactly $k$ segments. If $k' = 2c' + 1$, then $c' \leq c-1$, and by \cref{lemma:form}, $f'_2$ covers $g_2' = M+c'+1$ gold and $m_2' = M-c'$ mines. We have $g_2' < g_2$ and $m_2' > m_2$. Therefore by \cref{lemma:cover}, $u'_2 < \hat{u}_2$.

If $k' = 2c'$, then $c'\leq c$. By \cref{lemma:form}, $f'_2$ either covers $g_2' = M+c'+1$ gold and $m_2' = M-c'+1$ mines, in which case $g_2' \leq g_2$ and $m_2' > m_2$, or $g_2' = M+c'$ gold and $m_2' = M-c'$ mines, in which case $g_2' < g_2$ and $m_2' \geq m_2$. Therefore by \cref{lemma:cover}, $u'_2 < \hat{u}_2$.

\paragraph{1.2} By symmetry, the above proof also applies to the case where $f_1^*$ satisfies condition $S_2$ (symmetry with respect to inverting the direction of $x$ and $y$ axis).
\\
\\
\noindent\textbf{2. $k=2c$ (even $k$), $c\leq M$}

By \cref{lemma:form}, $f_1^*$ must satisfy condition $S_3$ or $S_4$.

\paragraph{2.1} If $f_1^*$ satisfies condition $S_3$, let $f^*_1 = \{[a_{\xi}^-, a_{\xi}^+]\}_{{\xi}=0}^{c-1}$
where $a_0^- = 0$,
and $\{a_{\xi}^+ = 4j_{\xi}^+\}_{\xi=0}^{c-1}$,
$\{a_{\xi}^- = 4j_{\xi}^-+3\}_{\xi=1}^{c-1}$, and $f_1^*$ covers $M+c+1$ gold and $M-c+1$ mines.
We construct $\hat{f}_2 = \{[\hat{a}_{\xi}^-, \hat{a}_{\xi}^+]\}_{{\xi}=0}^{c-1}$ according to $f_1^*$ by sequentially setting $\hat{a}_{0}^-, \hat{a}_{0}^+, \hat{a}_{1}^-, \hat{a}_{1}^+, \dots, \hat{a}_{c-1}^-, \hat{a}_{c-1}^+$ with
$\hat{a}_{0}^- = \max( 4j_{0}^+ - 1, 3)$,
$\{\hat{a}_{\xi}^- = \max( 4j_{\xi}^+ - 1, \hat{a}_{\xi-1}^+ + 3)\}_{\xi=1}^{c-1}$,
$\{\hat{a}_{\xi}^+ = \max( 4 \cdot j_{\xi+1}^-, \hat{a}_{\xi}^- + 1)\}_{\xi=0}^{c-2}$,
$\hat{a}_{c-1}^+ = 4M + 1$.
This $\hat{f}_2$ satisfies condition $S_4$ and covers $g_2 = M+c$ gold and $m_2 = M-c$ mines. In particular, among the gold covered by $\hat{f}_2$, $2c-1$ of them are also covered by $f_1^*$, and $M-c+1$ of them are covered by $\hat{f}_2$ only. Therefore, $\hat{f}_2$ achieves best coverage. $\hat{f}_2$ achieves a payoff of
\begin{equation*}
    \hat{u}_2 = M - c + 1 + (2c-1)\rho + (M - c)\mu.
\end{equation*}

We show here that any $f'_2 \in \mathcal{F}_{k'}$ where $k' < k$ will achieve a payoff $u'_2 < \hat{u}_2$, therefore $f_2^*$ must use exactly $k$ segments. If $k' = 2c' + 1$, then $c' \leq c - 1$, and by \cref{lemma:form}, $f'_2$ covers $g_2' = M+c'+1$ gold and $m_2' = M-c'$ mines. We have $g_2' \leq g_2$ and $m_2' > m_2$. Therefore by \cref{lemma:cover}, $u'_2 < \hat{u}_2$.

If $k' = 2c'$, then $c'\leq c-1$. By \cref{lemma:form}, $f'_2$ either covers $g_2' = M+c'+1$ gold and $m_2' = M-c'+1$ mines, in which case $g_2' \leq g_2$ and $m_2' > m_2$, or $g_2' = M+c'$ gold and $m_2' = M-c'$ mines, in which case $g_2' < g_2$ and $m_2' > m_2$. Therefore by \cref{lemma:cover}, $u'_2 < \hat{u}_2$.

\paragraph{2.2} If $f_1^*$ satisfies condition $S_4$, let $f^*_1 = \{[a_{\xi}^-, a_{\xi}^+]\}_{{\xi}=0}^{c-1}$
where $a_{c-1}^+ = 4M+1$,
and $\{a_{\xi}^+ = 4j_{\xi}^+\}_{\xi=0}^{c-2}$,
$\{a_{\xi}^- = 4j_{\xi}^-+3\}_{\xi=0}^{c-1}$, and $f_1^*$ covers $M+c$ gold and $M-c$ mines.
We construct $\hat{f}_2 = \{[\hat{a}_{\xi}^-, \hat{a}_{\xi}^+]\}_{{\xi}=0}^{c-1}$ according to $f_1^*$ by setting $\hat{a}_{0}^- = 0, \hat{a}_{0}^+ = 4j_0^-, \{\hat{a}_{\xi}^- = 4j_{\xi-1}^+ + 3, \hat{a}_{\xi}^+ = 4 \cdot  \max(j_{\xi}^-, j_{\xi-1}^+ + 1)\}_{\xi=1}^{c-1}$.
This $\hat{f}_2$ satisfies condition $S_3$ and covers $g_2 = M+c+1$ gold and $m_2 = M-c+1$ mines. In particular, among the gold covered by $\hat{f}_2$, $2c-1$ of them are also covered by $f_1^*$, and $M-c+2$ of them are covered by $\hat{f}_2$ only. Therefore, $\hat{f}_2$ achieves best coverage. $\hat{f}_2$ achieves a payoff of
\begin{equation*}
    \hat{u}_2 = M - c + 2 + (2c-1)\rho + (M - c + 1)\mu.
\end{equation*}

We show here that any $f'_2 \in \mathcal{F}_{k'}$ where $k' < k$ will achieve a payoff $u'_2 < \hat{u}_2$, therefore $f_2^*$ must use exactly $k$ segments. If $k' = 2c' + 1$, then $c' \leq c - 1$, and by \cref{lemma:form}, $f'_2$ covers $g_2' = M+c'+1$ gold and $m_2' = M-c'$ mines. We have $g_2' < g_2$ and $m_2' \geq m_2$. Therefore by \cref{lemma:cover}, $u'_2 < \hat{u}_2$.

If $k' = 2c'$, then $c'\leq c-1$. By \cref{lemma:form}, $f'_2$ either covers $g_2' = M+c'+1$ gold and $m_2' = M-c'+1$ mines, in which case $g_2' < g_2$ and $m_2' > m_2$, or $g_2' = M+c'$ gold and $m_2' = M-c'$ mines, in which case $g_2' < g_2$ and $m_2' \geq m_2$. Therefore by \cref{lemma:cover}, $u'_2 < \hat{u}_2$.
\\
\\
\noindent\textbf{3. $k=2c$ (even $k$), $c = M+1$}

By \cref{lemma:form}, $f_1^*$ must satisfy condition $S_3$ or $S_4$. In fact, in this case, no function satisfies $S_4$, and there is only one function satisfies $S_3$, which is the function that covers all $2M+2$ gold and no mine.
Construct $\hat{f}_2$ to be the same as $f_1^*$, which covers all $g_2 = 2M+2$ gold and $m_2=0$ mine. Since $f_1^*$ already covers all gold, $\hat{f}_2$ trivially achieves best coverage.

We show here that any $f'_2 \in \mathcal{F}_{k'}$ where $k' < k$ will achieve a payoff $u'_2 < \hat{u}_2$, therefore $f_2^*$ must use exactly $k$ segments. If $k' = 2c' + 1$, then $c' \leq c - 1$, and by \cref{lemma:form}, $f'_2$ covers $g_2' = M+c'+1$ gold and $m_2' = M-c'$ mines. We have $g_2' < g_2$ and $m_2' \geq m_2$. Therefore by \cref{lemma:cover}, $u'_2 < \hat{u}_2$.

If $k' = 2c'$, then $c'\leq c-1$. By \cref{lemma:form}, $f'_2$ either covers $g_2' = M+c'+1$ gold and $m_2' = M-c'+1$ mines, in which case $g_2' < g_2$ and $m_2' > m_2$, or $g_2' = M+c'$ gold and $m_2' = M-c'$ mines, in which case $g_2' < g_2$ and $m_2' \geq m_2$. Therefore by \cref{lemma:cover}, $u'_2 < \hat{u}_2$.
\\
\\
Now we have shown that if one of the strategies in a PNE uses exactly $k$ segments, then the other strategy must also use exactly $k$ segments. What is left to show is that there is no PNE where both strategies use less than $k$ segments.

We prove by contradiction. Assume there is a PNE $(f_1^*, f_2^*)$ where both $f_1^*$ and $f_2^*$ use less than $k$ segments. By the induction hypothesis, $f_1^*, f_2^* \in \mathcal{F}_{k-1}$.
\\
\\
\noindent\textbf{1. $k-1=2c+1$ (even $k$), $c \leq M$}

By \cref{lemma:form}, $f_1^*$ must satisfy condition $S_1$ or $S_2$. If $f_1^*$ satisfies condition $S_1$, let $f^*_1 = \{[a_{\xi}^-, a_{\xi}^+]\}_{{\xi}=0}^c$
where $a_0^- = 0, a_c^+=4M+1$,
$\{a_{\xi}^+ = 4j_{\xi}^+\}_{\xi=0}^{c-1}$,
$\{a_{\xi}^- = 4j_{\xi}^-+3\}_{\xi=1}^{c}$.
We construct $\hat{f}_2 = \{[\hat{a}_{\xi}^-, \hat{a}_{\xi}^+]\}_{{\xi}=0}^{c}$ according to $f_1^*$ by setting $\hat{a}_{0}^- = 0, \hat{a}_{0}^+ = 0, \{\hat{a}_{\xi}^- = 4j_{\xi-1}^+ + 3, \hat{a}_{\xi}^+ = 4 \cdot  \max(j_{\xi}^-, j_{\xi-1}^+ + 1)\}_{\xi=1}^{c}$. It is easy to check that $\hat{f}_2$ uses $k$ segments, covers $g_2 = M+c+2$ gold and $m_2 = M - c$ mines, and achieves best coverage. By \cref{lemma:form}, $f_2^*$ covers $g_2^* = M+c+1$ gold and $m_2^* = M-c$ mines. Thus $g_2^* < g_2$ and $m_2^* \geq m_2$. By \cref{lemma:cover}, $u_2^* < \hat{u}_2$, therefore $(f_1^*, f_2^*)$ cannot be a PNE, contradiction.

By symmetry, the above proof also applies to the case where $f_1^*$ satisfies condition $S_2$ (symmetry with respect to inverting the direction of $x$ and $y$ axis).
\\
\\
\noindent\textbf{2. $k-1=2c$ (odd $k$), $c \leq M$}

By \cref{lemma:form}, $f_1^*$ must satisfy condition $S_3$ or $S_4$. If $f_1^*$ satisfies condition $S_3$, let $f^*_1 = \{[a_{\xi}^-, a_{\xi}^+]\}_{{\xi}=0}^{c-1}$
where $a_0^- = 0$,
and $\{a_{\xi}^+ = 4j_{\xi}^+\}_{\xi=0}^{c-1}$,
$\{a_{\xi}^- = 4j_{\xi}^-+3\}_{\xi=1}^{c-1}$.
We construct $\hat{f}_2 = \{[\hat{a}_{\xi}^-, \hat{a}_{\xi}^+]\}_{{\xi}=0}^{c}$ according to $f_1^*$ by sequentially setting $\hat{a}_{0}^-, \hat{a}_{0}^+, \hat{a}_{1}^-, \hat{a}_{1}^+, \dots, \hat{a}_{c}^-, \hat{a}_{c}^+$ with
$\hat{a}_{0}^- = 0$, $\hat{a}_{0}^+ = 0$,
$\hat{a}_{1}^- = \max( 4j_{0}^+ - 1, 3)$,
$\{\hat{a}_{\xi}^- = \max( 4j_{\xi-1}^+ - 1, \hat{a}_{\xi-1}^+ + 3)\}_{\xi=2}^{c}$,
$\{\hat{a}_{\xi}^+ = \max( 4 \cdot j_{\xi}^-, \hat{a}_{\xi}^- + 1)\}_{\xi=1}^{c-1}$,
$\hat{a}_{c}^+ = 4M + 1$. It is easy to check that $\hat{f}_2$ uses $k$ segments, covers $g_2 = M+c+1$ gold and $m_2 = M - c$ mines, and achieves best coverage.
By \cref{lemma:form}, $f_2^*$ either covers $g_2^* = M+c+1$ gold and $m_2^* = M-c+1$ mines, in which case $g_2^* \leq g_2$ and $m_2^* > m_2$, or $g_2^* = M+c$ gold and $m_2^* = M-c$ mines, in which case $g_2^* < g_2$ and $m_2^* \geq m_2$.
Therefore by \cref{lemma:cover}, $u_2^* < \hat{u}_2$, which means $(f_1^*, f_2^*)$ cannot be a PNE, contradiction.

If $f_1^*$ satisfies condition $S_4$, let $f^*_1 = \{[a_{\xi}^-, a_{\xi}^+]\}_{{\xi}=0}^{c-1}$
where $a_{c-1}^+ = 4M+1$,
and $\{a_{\xi}^+ = 4j_{\xi}^+\}_{\xi=0}^{c-2}$,
$\{a_{\xi}^- = 4j_{\xi}^-+3\}_{\xi=0}^{c-1}$. $f_1^*$ covers $g_1^* = M+c$ gold and $m_1^* = M-c$ mines.
We construct $\hat{f}_1 = \{[\hat{a}_{\xi}^-, \hat{a}_{\xi}^+]\}_{{\xi}=0}^{c}$ according to $f_1^*$ by setting $\hat{a}_{0}^- = 0, \hat{a}_{0}^+ = 0,  \{\hat{a}_{\xi}^- = a_{\xi-1}^-, \hat{a}_{\xi}^+ = a_{\xi-1}^+\}_{\xi=1}^{c}$, i.e. $\hat{f}_1$ is identical to $f_1^*$ except $\hat{f}_1(0) = 1$. $\hat{f}_1$ uses $k$ segments, and covers exactly the same set of gold and mines as $f_1^*$ plus the gold at $t=0$. Therefore, $\hat{f}_1$'s payoff is strictly higher than $f_1^*$'s payoff. This means $(f_1^*, f_2^*)$ cannot be a PNE, contradiction.
\\
\\
This finishes the proof that there exists no PNE where both strategies use less than $k$ segments. We have also shown that if one strategy in a PNE uses $k$ segments, the other strategy must also use $k$ segments. This together shows that for all PNE, both strategies in the equilibria must use exactly $k$ segments. This finishes the proof by induction.

\end{proof}

\aogthm*

\begin{proof}
We consider different values of $b$.
\\
\\
\textbf{1. $b = 2c + 1, 0\leq c \leq M$}

By \cref{lemma:exact}, both $f_1^*$ and $f_2^*$ use exactly $b$ segments, i.e. $f_1^*, f_2^* \in \mathcal{F}_b$. By \cref{lemma:form}, both $f_1^*$ and $f_2^*$ must satisfy condition $S_1$ or $S_2$, and each covers $M+c+1$ gold and $M-c$ mines.
If $f_1^*$ satisfies $S_1$, denote $f^*_1 = \{[a_{\xi}^-, a_{\xi}^+]\}_{{\xi}=0}^c$
where $a_0^- = 0, a_c^+=4M+1$,
$\{a_{\xi}^+ = 4j_{\xi}^+\}_{\xi=0}^{c-1}$,
$\{a_{\xi}^- = 4j_{\xi}^-+3\}_{\xi=1}^{c}$.
Same as in the proof of \cref{lemma:exact}, we construct $\hat{f}_2 = \{[\hat{a}_{\xi}^-, \hat{a}_{\xi}^+]\}_{{\xi}=0}^{c-1}$ according to $f_1^*$ by setting $\{\hat{a}_{\xi}^- = 4j_{\xi}^+ + 3, \hat{a}_{\xi}^+ = 4 \cdot  \max(j_{\xi+1}^-, j_{\xi}^+ + 1)\}_{\xi=0}^{c-1}$. $\hat{f}_2\in \mathcal{F}_b$ achieves best coverage and a payoff of
$\hat{u}_2 = M - c + 1 + 2c\rho + (M - c)\mu.$
Since $f_2^*$ always covers $M+c+1$ gold and $M-c$ mines, by \cref{lemma:payoff}, $f_2^*$'s payoff $u_2^* \leq \hat{u}_2$. But by definition of Nash equilibrium, $u_2^* \geq \hat{u}_2$. Therefore, $u_2^*=\hat{u}_2$, i.e. all $f_2^*$ must achieve the same payoff of $M - c + 1 + 2c\rho + (M - c)\mu$.

By symmetry (with respect to inverting the direction of $x$ and $y$ axis), the above proof can also be applied to show that if $f_1^*$ satisfies $S_2$, then all $f_2^*$ must achieve the same payoff of $M - c + 1 + 2c\rho + (M - c)\mu$.

Therefore, in all cases, $u_2^* = M - c + 1 + 2c\rho + (M - c)\mu$. Similarly, $u_1^* = M - c + 1 + 2c\rho + (M - c)\mu$. So
\begin{align*}
    W_{\equil}(b) &= 2M(1+\mu) + 2 + 2(2\rho - \mu - 1)c \\
    &= 2M(1+\mu) + 2 + (2\rho - \mu - 1)(b-1) \\
    &= (2M + 1)(1 + \mu) + 2(1-\rho) + (2\rho - \mu - 1)b.
\end{align*}

\noindent\textbf{2. $b = 2c, 1\leq c \leq M$}

By \cref{lemma:exact}, $f_1^*, f_2^* \in \mathcal{F}_b$. By \cref{lemma:form}, both $f_1^*$ and $f_2^*$ must satisfy condition $S_3$ or $S_4$.
If $f_1^*$ satisfies $S_3$, denote $f^*_1 = \{[a_{\xi}^-, a_{\xi}^+]\}_{{\xi}=0}^{c-1}$
where $a_0^- = 0$,
and $\{a_{\xi}^+ = 4j_{\xi}^+\}_{\xi=0}^{c-1}$,
$\{a_{\xi}^- = 4j_{\xi}^-+3\}_{\xi=1}^{c-1}$, and $f_1^*$ covers $M+c+1$ gold and $M-c+1$ mines.
Same as in the proof of \cref{lemma:exact}, we construct $\hat{f}_2 = \{[\hat{a}_{\xi}^-, \hat{a}_{\xi}^+]\}_{{\xi}=0}^{c-1}$ according to $f_1^*$ by sequentially setting $\hat{a}_{0}^-, \hat{a}_{0}^+, \hat{a}_{1}^-, \hat{a}_{1}^+, \dots, \hat{a}_{c-1}^-, \hat{a}_{c-1}^+$ with
$\hat{a}_{0}^- = \max( 4j_{0}^+ - 1, 3)$,
$\{\hat{a}_{\xi}^- = \max( 4j_{\xi}^+ - 1, \hat{a}_{\xi-1}^+ + 3)\}_{\xi=1}^{c-1}$,
$\{\hat{a}_{\xi}^+ = \max( 4 \cdot j_{\xi+1}^-, \hat{a}_{\xi}^- + 1)\}_{\xi=0}^{c-2}$,
$\hat{a}_{c-1}^+ = 4M + 1$.
This $\hat{f}_2$ satisfies condition $S_4$ and achieves a payoff of $\hat{u}_2 = M - c + 1 + (2c-1)\rho + (M - c)\mu$.
For any $f'_2$ that satisfies $S_3$, by \cref{lemma:form}, it covers $M+c+1$ gold and $M-c+1$ mines. By \cref{lemma:payoff}, such $f'_2$'s payoff
\begin{align*}
    u'_2 &\leq (1-\rho)(M-c+1) + \rho(M+c+1) +\mu(M-c+1) \\
    &= M - c + 1 + 2c\rho + (M - c+1)\mu \\
    &= \hat{u}_2 +\rho + \mu < \hat{u}_2.
\end{align*}
By the definition of Nash equilibrium, $u_2^* \geq \hat{u}_2$, so $f_2^*$ cannot satisfy $S_3$. Therefore, $f_2^*$ must satisfy $S_4$, and by \cref{lemma:form}, $f_2^*$ covers $M+c$ gold and $M-c$ mines. So by \cref{lemma:payoff}, $u_2^* \leq \hat{u}_2$. Therefore, $u_2^* = \hat{u}_2$, i.e. $f_2^*$ always satisfies $S_4$ and achieves a payoff of $u_2^* = M - c + 1 + (2c-1)\rho + (M - c)\mu$.

If $f_1^*$ satisfies $S_4$, denote $f^*_1 = \{[a_{\xi}^-, a_{\xi}^+]\}_{{\xi}=0}^{c-1}$
where $a_{c-1}^+ = 4M+1$,
and $\{a_{\xi}^+ = 4j_{\xi}^+\}_{\xi=0}^{c-2}$,
$\{a_{\xi}^- = 4j_{\xi}^-+3\}_{\xi=0}^{c-1}$, and $f_1^*$ covers $M+c$ gold and $M-c$ mines.
We construct $\hat{f}_2 = \{[\hat{a}_{\xi}^-, \hat{a}_{\xi}^+]\}_{{\xi}=0}^{c-1}$ according to $f_1^*$ by setting $\hat{a}_{0}^- = 0, \hat{a}_{0}^+ = 4j_0^-, \{\hat{a}_{\xi}^- = 4j_{\xi-1}^+ + 3, \hat{a}_{\xi}^+ = 4 \cdot  \max(j_{\xi}^-, j_{\xi-1}^+ + 1)\}_{\xi=1}^{c-1}$.
This $\hat{f}_2$ satisfies condition $S_3$ and achieves a payoff of $\hat{u}_2 = M - c + 2 + (2c-1)\rho + (M - c + 1)\mu$.
For any $f'_2$ that satisfies $S_4$, by \cref{lemma:form}, it covers $M+c$ gold and $M-c$ mines. Denote $d$ as the number of gold that is covered by both $f_1^*$ and $f'_2$, noting that both $f_1^*$ and $f'_2$ cannot cover the gold at $t=0$ and $t=4M+1$, we have $M + c + M + c - d \leq 2M$, so $d\geq 2c$. Therefore, such $f'_2$'s payoff
\begin{align*}
    u'_2 &= M + c - d + d\rho + (M-c)\mu\\
    &\leq M + c - 2c(1-\rho) + (M-c)\mu \\
    &= \hat{u}_2 - 2 + \rho - \mu < \hat{u}_2.
\end{align*}
By the definition of Nash equilibrium, $u_2^* \geq \hat{u}_2$, so $f_2^*$ cannot satisfy $S_4$. Therefore, $f_2^*$ must satisfy $S_3$, and by \cref{lemma:form}, $f_2^*$ covers $M+c+1$ gold and $M-c+1$ mines. So by \cref{lemma:payoff}, $u_2^* \leq \hat{u}_2$. Therefore, $u_2^* = \hat{u}_2$, i.e. $f_2^*$ always satisfies $S_3$ and achieves a payoff of $u_2^* =  M - c + 2 + (2c-1)\rho + (M - c + 1)\mu$.

Combining the above results, we can show that for any PNE $(f_1^*, f_2^*)$, one of $f_1^*$ and $f_2^*$ must satisfy $S_3$ and achieves a payoff of $M - c + 2 + (2c-1)\rho + (M - c + 1)\mu$, and the other must satisfy $S_4$ and achieves a payoff of $M - c + 1 + (2c-1)\rho + (M - c)\mu$. Therefore,
\begin{align*}
    W_{\equil}(b) &= (2M + 1)(1 + \mu) + 2(1-\rho) + 2(2\rho - \mu - 1)c \\
    &= (2M + 1)(1 + \mu) + 2(1-\rho) + (2\rho - \mu - 1)b.
\end{align*}

\noindent\textbf{3. $b \geq 2M+2$}

Since $f_1^*$ and $f_2^*$ can have at most $2M+2$ segments, when $b \geq 2M+2$, $f_1^*, f_2^* \in \mathcal{F}_{2M+2}$. There is only one function in $\mathcal{F}_{2M+2}$, which is the function that covers all gold and no mines, therefore both $f_1^*$ and $f_2^*$ must be this particular function. So $u_1^* = u_2^* = (2M+2)\rho$, $W_{\equil}(b) = (4M+4)\rho$.

\end{proof}

\subsection{Best social welfare by centralized solutions}
\label{append:best}

\begin{proposition}
    For the \aog{}, the best social welfare achieved by any centralized solution under $\strategyLevel{b}$ is
    \begin{equation*}
        W_{\best}(b) =
        \begin{cases}
        2M+2 + (2M+1)\mu - \mu b & \text{if } b \leq 2M+1 \\
        2M+2 & \text{if } b \geq 2M+2
        \end{cases}.
    \end{equation*}
\end{proposition}

\begin{proof}
First, we notice that a function with $c$ changes from $y=0$ to $y=1$ covers at least $M-c$ mines.

For $b=2c, c\leq M$, we can construct two functions $(\hat{f}_1, \hat{f}_2)$ achieving $W_{\textrm{best}}(b)$. $\hat{f}_1$ starts from $y=1$ and ends at $y=0$, covers all gold on $y=1$ except the rightmost one, and the rightmost gold on $y=0$; $\hat{f}_2$ starts from $y=0$ and ends at $y=1$ covers all gold on $y=0$ except the rightmost one, and the rightmost gold on $y=1$. Both use the line changes to avoid mines, with $\hat{f}_1$ covering $M-c+1$ mines and $\hat{f}_2$ covering $M-c$ mines. No other functions can achieve a better welfare, since they cannot jointly cover more gold, and $w_g = xr_g(x)$ attains its maximum at $x=1$. And to jointly cover fewer mines, both functions need to start from $y=0$ and ends at $y=1$, which can reduce the number of mines covered by at most 1. But then the two functions can only jointly cover at most $2M$ gold, which makes the social welfare lower than that of $(\hat{f}_1, \hat{f}_2)$.

For $b=2c+1, c\leq M$, we can construct two functions $(\hat{f}_1, \hat{f}_2)$ achieving $W_{\textrm{best}}(b)$. They jointly cover all gold with no overlap, and $M-c$ mines each. The construction is simply let $\hat{f}_1$ covers all gold on $y=1$, $\hat{f}_2$ covers all gold on $y=0$. No other functions can achieve a better welfare, since they cannot jointly cover more gold or less mines.

For $b \geq 2M+2$, let $\hat{f}_1$ covers all gold on $y=1$ and no mines, and $\hat{f}_2$ covers all gold on $y=0$ and no mines. This achieves $W_{best}(b)$, which is in fact the maximum possible social welfare of this game.

\end{proof}

\subsection{Necessary condition for all PNEs having the same social welfare for any $b$ and $M$}
\label{append:aog-nec}

\begin{proposition}
$-2+\rho < \mu < -\rho$ is also a necessary condition for all PNEs having the same social welfare for any $b$ and $M$.
\end{proposition}

\begin{proof}
We provide constructions showing that if this condition is not satisfied, there is always some $b$ and $M$ where different PNEs have different social welfare.
If $\mu \geq -\rho$, for $M > \frac{2(\rho+\mu)}{1-\rho} + 1$ and $b=2$, $(\{[0,0]\}, \{[0,4M]\})$ is a PNE, $(\{[0, 4\cdot \lfloor M/2 \rfloor]\}, \{[4\cdot \lfloor M/2 \rfloor + 3, 4M+1]\})$ is another PNE, and their social welfare is different.
If $\mu \leq -2 + \rho$, for $M > \frac{2(-\mu-\rho)}{1-\rho}$ and $b=2$, $(\{[4M-1,4M+1]\}, \{[3,4M+1]\})$ is a PNE, $(\{[4\cdot \lfloor M/2 \rfloor+3, 4M+1]\}, \{[0, 4\cdot \lfloor M/2 \rfloor]\})$ is a PNE, and their social welfare is different.
\end{proof}
% f}}}

% vim: tw=80 filetype=tex foldmethod=marker foldmarker=f{{{,f}}} spell spelllang=en

\end{document}